\documentclass[11pt]{article}

\usepackage[margin=1in]{geometry}
\usepackage{times}
\usepackage{microtype}
\usepackage{amsmath,amssymb,amsthm}
\usepackage{thmtools,thm-restate}
\usepackage{graphicx}
\usepackage{subcaption}
\usepackage{booktabs}
\usepackage{hyperref}
\usepackage{cleveref}
\usepackage[numbers]{natbib}
\usepackage{xcolor}
\usepackage{algorithm}
\usepackage{algpseudocode}
\usepackage{parskip}

\newtheorem{theorem}{Theorem}
\newtheorem{lemma}[theorem]{Lemma}

\title{Permutation-Based Stegomalware in Large Language Models: Threats and Countermeasures}
\author{
  Danny Wood \\
  Fuzzy Labs \\
  \texttt{danny@fuzzylabs.ai}
  \and
  James Stringer \\
  Fuzzy Labs \\
  \texttt{james@fuzzylabs.ai}
}
\date{}

\begin{document}

\maketitle

\begin{abstract}

  The difficulty of training large language models (LLMs), together with their ubiquity, raises
  the threat of stegomalware, where malicious payloads are embedded into model weights.
  Recent work has demonstrated the use of permutation symmetry in model weights to
  mitigate these threats, but failed to show neutralization of stegomalware across
  all weights for LLMs.
  In this paper, we demonstrate the full potential of behavior-preserving symmetries
  as a defense against stegomalware, as well as the risks these symmetries pose when
  exploited by attackers.

  For stegomalware neutralization, we improve upon previous work, demonstrating that it is possible
  to select permutations which displace all model parameters. This contrasts with
  previous methods which left a significant percentage of weights
  unaltered in LLMs.
  When used in an attack, we show that permutation symmetries can encode malware into
  the weights of a model in a way that is theoretically lossless, requires no
  retraining after encoding, and needs no payload-specific information in the extraction
  script---a combination of characteristics not previously seen in any single method.

  While theoretically lossless, permutation can in practice alter model
  behavior due to the accumulation of numerical error. We therefore quantify the loss in
  model performance associated with applying these methods, for both attack and defense,
  showing it to be minimal.
\end{abstract}

\section{Introduction}
\label{sec:intro}

Large language models are an increasingly vital part of many software products, but
for most academic, personal, or even industrial uses, training a custom large language model
from scratch is infeasible. This means that most developers and end users rely on pre-trained
models from online repositories such as Hugging Face. This therefore creates a supply
chain risk, as malicious actors can use these models as an attack vector.

Attacks on the supply chain for machine learning tools and software are becoming
increasingly prevalent, both through conventional means such as compromising trusted
packages~\citep{dholakia_security_2026} and more AI-specific means such as serving compromised
models on Hugging Face.
This kind of supply chain attack was listed as one of OWASP's top ten risks for LLMs and
GenAI Apps~\citep{owasp_owasp_2025} and is likely to increase in prevalence as the demand
for AI in products grows.

Real world attacks using compromised models thus far typically hide malicious payloads within
code and data separate from the model weights themselves.
However, recent research has demonstrated that it is possible for attackers to
make their malicious code harder to detect through the use of steganography,
creating malware which is hidden within the structures of the neural network
itself~\citep{liu_stegonet_2020, wang_evilmodel_2022, hitaj_maleficnet_2022}.
Existing work has shown that neural networks
make appealing targets for this kind of malware, since the high entropy and high fault
tolerance of neural network weights allow for the insertion of significantly sized payloads
with little drop in model performance.
Steganographic malware has already proven both feasible and capable of evading
detection on other online software hubs, such as smartphone app
stores~\citep{lin_stegomalware_2015}, suggesting a need to understand
and address these vulnerabilities in this new domain.

Recent work has begun to develop defenses against these attacks by exploiting functional
invariances in neural network architectures, shuffling weights to disrupt payload
extraction in specific constructions~\citep{torpmann-hagen_defending_2025, gilkarov_neuperm_2025}.
This kind of functional invariance has also been used to watermark models for intellectual-property
protection~\citep{fernandez_functional_2024} and to analyze the structure of machine
learning models~\citep{ainsworth_git_2022}.

As we will show, these invariances also pose a threat: they can be used not
only to neutralize stegomalware but also to embed it. Embedding a payload in a large
language model this way has significant benefits for the attacker: in exact arithmetic
it preserves model behavior completely, requires no retraining, and needs no
payload-specific extraction script. Existing attacks in the literature all sacrifice at least one of
these three properties.

Fortunately, the attack and countermeasures in this scenario are complementary. While
we are able to demonstrate that it is possible to store significant amounts of
information in a large language model, we also show that the attack can be neutralized
by randomly permuting the weight matrices of the network.
Previous work relies on permutations whose effects are localized, confined within individual
components or layers. We show that, as well as these localized
symmetries, there are \emph{global} permutations that propagate consistently
across every layer. Composing the two, along with careful selection of the permutation to apply,
shuffles 100\% of the parameters in an
LLM, which is sufficient for disruption of all known attacks in the
literature.

Our contributions are as follows:
\begin{itemize}
  \item We show how \emph{global} permutation symmetries, propagating across the whole
  network, can be exploited to neutralize stegomalware in 100\% of the parameters of a
  large language model, beating the previous state of the art which managed less than
  two thirds of all parameters.
  \item We introduce PermaNet, a novel, (theoretically) lossless method of embedding hidden
  information in neural networks.
  \item We demonstrate the efficacy of this neutralization, showing complete disruption
   of payload extraction on existing stegomalware encoding techniques, and complete removal of PermaNet-encoded payloads.
  \item We test our methods of malware embedding and neutralization, verifying a minimal reduction
  in model performance. This is less than the effects of quantization levels considered to be `near lossless',
  and disappears as the level of numerical precision at inference is increased.
\end{itemize}

\section{Threat Model}
\label{sec:background}
\label{sec:threatmodel}
The threat model we define aligns closely with those described in~\citep{liu_stegonet_2020}
and~\citep{gilkarov_neuperm_2025}; this model is representative of how threats are likely to surface
when embedded in pre-trained neural networks distributed to consumers through third party
distribution platforms like Hugging Face. We define the following personas.

\textbf{End user.} The end user of this attack is any consumer of pre-trained neural networks which
are distributed through third party platforms. We do not restrict the end user to a non-expert
persona, as ML experts frequently consume pre-trained models for experimentation or deployment.
The end user may either be an independent actor or part of a larger organization,
and may deploy the model either on a local device or in a cloud environment. We assume
that the end user stores and deploys the model as is, without performing additional transformations
such as fine-tuning or quantization.

\textbf{Adversary.} We define the adversary as a third party provider of malicious pre-trained
models. Unofficial providers are common on distribution platforms and among them malicious
providers are largely indistinguishable from benign ones.
The adversary creates a malicious
neural network and advertises it on a third party distribution platform. They may associate
the malicious model with a trusted provider, such as through claiming their malicious model
as a fine-tune or quantization of a trusted model, or by masquerading as a trusted
organization, creating malicious spoofs of popular models. The goal of the adversary is to achieve
remote code execution or a similar exploit on the end user's machine using a payload
embedded in their malicious model. They must fulfill several conditions in order to achieve
this goal:
\begin{enumerate}
  \item Maintain the performance of the malicious model relative to a benign model, both in
  order to evade detection of malicious behavior and to avoid reduced desirability of a
  degraded model.
  \item Avoid detection with traditional anti-malware tools using methods like
  encoded string detection.
  \item Decode and execute the malicious payload during a stage in the normal deployment of
  the model, without any additional access to the end user's machine. Typically, this will
  occur during the deserialization of the model, such as through depickling.
\end{enumerate}
We suppose that, in order to fulfill these conditions, the adversary uses the PermaNet
method described in this work to store a malicious payload in the ordering of their
model's weights. PermaNet is the steganographic technique rather than a complete attack;
we consider trigger mechanisms to be outside the scope of this work. As in prior work, we assume that a small
\emph{bootstrap} routine runs when the model is loaded.
There have been several real-world vulnerabilities that could be used for this purpose,
for instance through the
well-known pickle deserialization vector, or in code packaged with the model weights.
The bootstrap reads the model's own weight tensors, runs the PermaNet
decoder to recover the payload bytes, and executes them.
PermaNet's bootstrap consists of payload-agnostic tensor-manipulation operations, carrying
no malicious bytes of its own. This attracts less suspicion than the payload it
reconstructs, meaning that signature-based scanning will not detect the payload,
only more benign-looking tensor operations.
We discuss this in more detail in Section~\ref{sec:discussion}.

In~\citet{lin_stegomalware_2015}, three types of stegomalware are described, depending
upon whether the payload, extraction method and cryptographic key for decoding are present
on the target system. The attacks and countermeasures which we describe in this paper
are agnostic to which scenario is considered. PermaNet is not intended as
a full end-to-end attack; instead, it is a specific steganographic method that an attacker
might employ on neural network architectures, following in the spirit of other
neural network steganography research like MaleficNet~\cite{hitaj_maleficnet_2022}.

\section{Permutation-Based Attacks}
\label{sec:method}

In this section, we show how it is possible to embed a potentially malicious payload
into a large language model by exploiting permutation symmetries in the model
parameters. Our argument has the following structure:
\begin{enumerate}
  \item We show that information can be encoded by reordering the rows or
  columns of a given matrix.
  \item We demonstrate how invariances in the structure of an LLM allow for permutation
  of parameter matrices while preserving model functionality.
  \item We show how the two points above can be combined to store a payload in the weights
  of an LLM while perfectly preserving model functionality---assuming perfect precision arithmetic---and
  without requiring that we know the original order of the rows/columns.
  \item Finally, we calculate how much data this allows to be encoded in real-world
  models, showing that it is substantial enough to encode common malware variants.
\end{enumerate}

\subsection{Encodable Bits}
For a matrix with $n$ unique rows, there are $n!$ possible orderings.
Given complete freedom to permute those rows as we like, it is therefore possible to encode
$\lfloor \log_2 n! \rfloor$ bits of information into that ordering.
Intuitively, for a matrix with $n\leq 26$ rows, we label the rows with
a canonical ordering as $A, B, C, \ldots$. We may then list all possible permutations
in alphabetical order, assigning each permutation a number between $0$ and $n!-1$ based on its
position in that list. This allows us to communicate that number to a recipient who knew the
matrix's original order by sending them the matrix whose rows are arranged according to the 
correspondingly ranked permutation.
Equivalently, this can be thought of as encoding bits in the binary representation
of the number, including leading zeros. This process is shown in Figure \ref{fig:encode-decode}.

\begin{figure}[ht]
\centering
\begin{subfigure}[b]{0.48\textwidth}
  \centering
  \includegraphics[width=\linewidth]{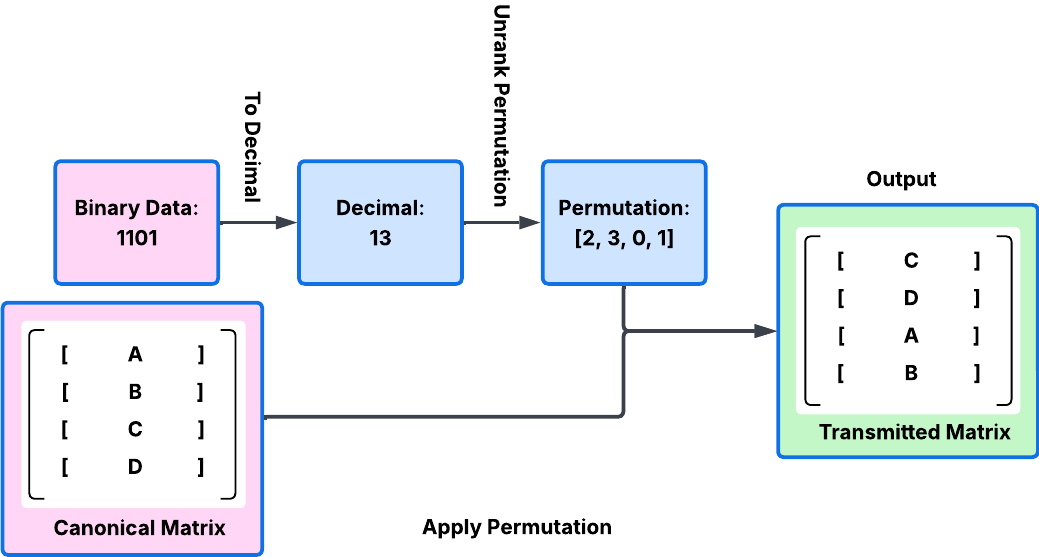}
  \caption{Encoding.}
  \label{fig:encoding}
\end{subfigure}
\hfill
\begin{subfigure}[b]{0.48\textwidth}
  \centering
  \includegraphics[width=\linewidth]{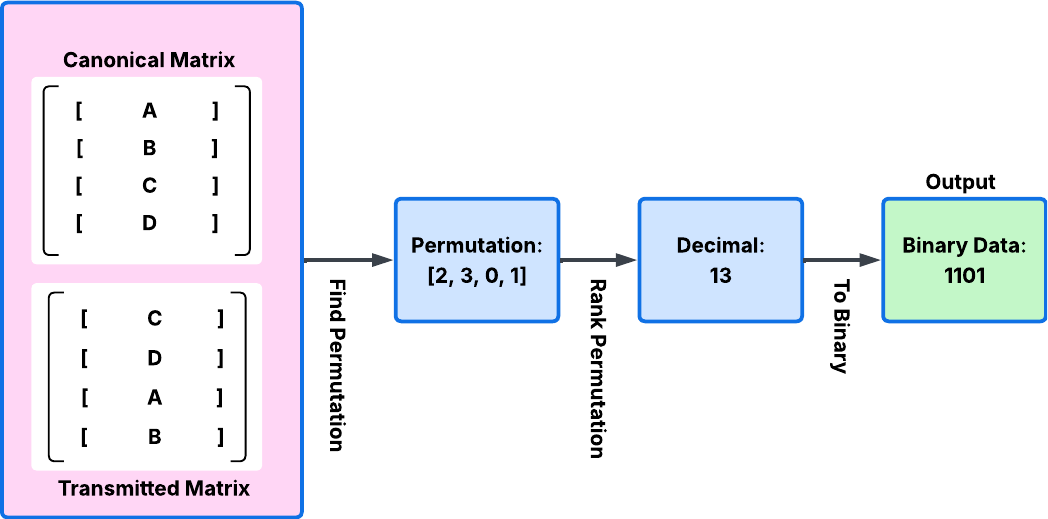}
  \caption{Decoding.}
  \label{fig:decoding}
\end{subfigure}
\caption{Visualization of the method for encoding and decoding information in the
ordering of the rows of a matrix. Left: Binary data is converted into a numerical value, then
encoded as a permutation. That permutation is then applied to the rows of the matrix before
transmission. Right: The recipient calculates the permutation required to transform the
canonical matrix into the received one. They then calculate the permutation's rank and
use this to reconstruct the binary data.}
\label{fig:encode-decode}
\end{figure}

In practice, the recipient does not need to know the original matrix beforehand;
they only need to have previously agreed upon a procedure to choose a canonical ordering of rows of a matrix, 
 that is, a method for constructing a bijection between the matrix's rows and the set $\{0, 1, \ldots, n-1\}$.
The sender and recipient also need efficient ways of calculating the rank of a given
permutation in the list of all possible $n!$ permutations, as well as the ability to
carry out the inverse of that operation.

For the canonical ordering, we deterministically order the rows: for instance by hashing
them, or by sorting each row's entries and comparing rows at their first differing entry.
We prefer the latter, as it allows us to establish a canonical ordering of the rows
which is invariant to permutations of the columns, allowing us to encode information
in both orderings at once.

Although the rank and its inverse can be computed using the lexicographical
ordering, we instead use a different ordering that is simpler to rank and unrank in
linear time~\citep{myrvold_ranking_2001}. Further to this, we modify this algorithm to derive
a non-recursive variant, allowing us to handle large matrices without reaching Python's
recursion-depth limit, as detailed in Appendix~\ref{app:rank-unrank}.

Calculating the information storage capacity of the permutation of $n$ objects requires
computing $\log_2 n!$. In our preliminary investigation, we encountered issues converting $n!$ to a float
for this operation, due to its size. Instead, we make use of the approximation
$$ \log_2 n! \approx \frac{\frac{1}{2}\ln 2 \pi n + n \ln n - n + \frac{1}{12n} - \frac{1}{360n^3}}{\ln 2}, $$
\noindent which is a variation on a well-known bound for $n!$.
Furthermore, we can verify that this approximation costs at most one bit. The proofs of
all results in this section are deferred to Appendix~\ref{app:proofs}.

\begin{restatable}{proposition}{propEncodableBits}
\label{prop:encodable-bits}
  For $n \geq 1$, we have
\begin{align*}
  \lfloor \log_2 n! \rfloor - 1 \leq \left\lfloor \frac{\frac{1}{2}\ln 2 \pi n + n \ln n - n + \frac{1}{12n} - \frac{1}{360n^3}}{\ln 2} \right\rfloor \leq  \lfloor \log_2 n! \rfloor.
\end{align*}
\end{restatable}

Hence, using this approximation in place of the true value costs at
most a single bit of storage. We conjecture that for $n>2$ the lower bound could be tightened and that
$\log_2 n!$ always has the same integer part as our approximation.
Through elementary computation, it can be seen that this grows at a rate of $\mathcal{O}(n \log n)$,
so the number of bits that can be stored grows super-linearly with the number of
rows of a matrix that can be permuted.

\subsection{Exploiting Invariances in a Transformer}

We will now show that there is an invariance in standard transformer architectures which
allows for permutations to be applied to all parameter matrices in the model while
preserving the model's function.

Before showing the full case of transformer models, we consider how such an invariance
works in a single-layer MLP. Recall that for an $m \times n$ matrix $W$, permutation of rows can be achieved by construction of
a permutation matrix $P$, with the property that $P^{-1} = P ^ \intercal$.

For input $x \in \mathbb{R}^d$:
\begin{align}
  \mathrm{MLP}(x;\, W_1, W_2) &= W_2\, \mathrm{ReLU}(W_1 x), \label{eq:mlp}
\end{align}
where $W_1 \in \mathbb{R}^{m \times d}$, $W_2 \in \mathbb{R}^{d \times m}$, and $\mathrm{ReLU}(z) = \max(0, z)$ applied elementwise.
Let $P \in \mathbb{R}^{m \times m}$ be a permutation matrix and define $W_1' = P W_1$ and $W_2' = W_2 P^{\intercal}$. Using the
fact that element-wise functions commute with permutations of elements in a vector, we get
\begin{align*}
  \mathrm{MLP}(x;\, W_1', W_2') &= W_2'\, \mathrm{ReLU}(W_1' x) = W_2 P^\intercal\, \mathrm{ReLU}(P W_1 x)\\
  &= W_2 P^\intercal P\, \mathrm{ReLU}(W_1 x) = W_2\, \mathrm{ReLU}(W_1 x) = \mathrm{MLP}(x;\, W_1, W_2).
\end{align*}
That is, applying the permutation $P$ and its transpose inside the network in the right
way preserves the network's behavior. This gives us a family of $m!$ MLPs which
are functionally equivalent but have different parameters. Like most MLP layers in modern
LLMs, our example does not include bias terms, but the same principle would apply to a network
including biases if we apply the same permutation to biases of the first layer.

The use of this simple example has been demonstrated
in~\citep{gilkarov_neuperm_2025, torpmann-hagen_defending_2025} to neutralize malware
hidden in the weights of MLPs. We will show how a more comprehensive permutation,
first observed in \citet{fernandez_functional_2024},
can be used to permute all the weights in an LLM while preserving model functionality.
We will also show extension of this method to modern mixture-of-experts (MoE) architectures.

To start, we show a set of permutations of parameters for each self-attention
and MLP layer of an LLM architecture
that results in equivariance: i.e., permutation of the inputs to the layer results
in an output that is equivalent to the original output up to the same permutation. We
will then demonstrate that applying that same permutation to the input and output embedding
weights results in invariance of the overall architecture.

We consider a multi-head self-attention on a sequence of length $n$ with embedding
dimension $d$.
Let $X \in \mathbb{R}^{n \times d}$ be the embeddings of the input sequence,
an attention head performs the function:
\begin{align}
  \textnormal{Att}(X;\, \theta_{\text{att}}) &= \textnormal{Softmax}\!\left(\frac{XW_Q W_K^\intercal X^\intercal}{\sqrt{d_k}}\right)XW_V \label{eq:att},
\end{align}
where $W_Q \in \mathbb{R}^{d \times d_k}$, $W_K \in \mathbb{R}^{d \times d_k}$, $W_V \in \mathbb{R}^{d \times d_v}$ are
the respective query, key, and value weight matrices, for some query and value dimension
values $d_k$ and $d_v$.
We write $\theta_{\text{att}} = \{W_Q, W_K, W_V, W_O\}$ for the full set of attention weight matrices,
where $W_O \in \mathbb{R}^{h d_v \times d}$ is the output projection matrix.
The values of the outputs of multiple heads are concatenated
and projected back down to the embedding dimension:
\begin{align}
  \textnormal{MHA}(X;\, \theta_{\text{att}}) &= \textnormal{Concat}(\textnormal{head}_1, \ldots, \textnormal{head}_h)\,W_O \label{eq:mha}\\
  \textnormal{head}_i &= \textnormal{Att}\!\left(X;\, \theta_{\text{att}}^{(i)}\right) \notag
\end{align}

Let $P \in \mathbb{R}^{d \times d}$ be a permutation matrix and define $W_Q' := P^\intercal W_Q$,
$W_K' := P^\intercal W_K$, $W_V' := P^\intercal W_V$ and $W_O' := W_O P$. We write $\theta'_{\text{att}} = \{W_Q', W_K', W_V', W_O'\}$\\

\begin{restatable}{proposition}{propMHA}
\label{prop:mha}
  Multi-head attention is equivariant under permutation of input features, i.e.,
  $\textnormal{MHA}(XP;\, \theta'_{\text{att}}) = \textnormal{MHA}(X;\, \theta_{\text{att}})P$.
\end{restatable}

In our definition of multi-head attention, we omit Rotary Position Embedding (RoPE)~\citep{su2024roformer},
which applies rotations to the query and key activations to encode sequence position. This choice is for
simplicity of exposition, and does not affect the proof since the rotation is on the internal dimensions of
the query and key vectors, not on the permuted embedding dimension.

Next we consider the MLP layers. These can be defined as in our example above, but
more typically in modern neural networks, a gated MLP is used of the form:
\begin{align}
  \textnormal{MLP}(X;\, \theta_{\text{mlp}}) &= \bigl(\sigma(X W_{\text{gate}}) \odot X W_{\text{up}}\bigr) W_{\text{down}} \label{eq:gmlp}
\end{align}
where $W_{\text{gate}} \in \mathbb{R}^{d \times d_{\text{ff}}}$, $W_{\text{up}} \in \mathbb{R}^{d \times d_{\text{ff}}}$, $W_{\text{down}} \in \mathbb{R}^{d_{\text{ff}} \times d}$.
$\sigma$ is some element-wise non-linear function, such as ReLU or GELU.
Which non-linear activation function is not relevant, it only matters that it
is applied elementwise, as is standard in all the most common LLM architectures.
We write $\theta_{\text{mlp}} = \{W_{\text{gate}}, W_{\text{up}}, W_{\text{down}}\}$.

Let $P \in \mathbb{R}^{d \times d}$ be a permutation matrix and define $W_{\text{gate}}' := P^\intercal W_{\text{gate}}$,
$W_{\text{up}}' := P^\intercal W_{\text{up}}$ and $W_{\text{down}}' := W_{\text{down}} P$.
We write $\theta'_{\text{mlp}} = \{W_{\text{gate}}', W_{\text{up}}', W_{\text{down}}'\}$.
This gives us the following proposition for an equivariance in gated MLP layers.\\

\begin{restatable}{proposition}{propGatedMLP}
\label{prop:gatedmlp}
  For gated MLPs, we have the equivariance
  $\textnormal{MLP}(XP;\, \theta'_{\text{mlp}}) = \textnormal{MLP}(X;\, \theta_{\text{mlp}})P.$
\end{restatable}

Additionally, before each self-attention layer and MLP layer, there is a per-feature normalization layer.
In the case of the popular family of open-weight Llama models~\citep{grattafiori_llama_2024}, this is RMSNorm. This also requires a permutation.
\begin{align}
  \textnormal{Ln}(x;\, \gamma) &= \gamma \odot \frac{x}{\sqrt{\textnormal{RMS}^2(x) + \varepsilon}} \label{eq:rmsnorm}
\end{align}
with $\textnormal{RMS}^2(x) = \frac{1}{d}\sum_{i=1}^{d} x_i^2$, $\gamma \in \mathbb{R}^d$, $\gamma' := \gamma P$.\\

\begin{restatable}{proposition}{propRMSNorm}
\label{prop:rmsnorm}
  For RMSNorm, we have the equivariance
  $\textnormal{Ln}(xP;\, \gamma') = \textnormal{Ln}(x;\, \gamma)\, P.$
\end{restatable}

Letting $W_E \in \mathbb{R}^{|V| \times d}$ be the embedding matrix and $z^0 = \mathbf{t} W_E$ where $\mathbf{t}$ is the one-hot token sequence, we define:
\begin{align}
  h^l &= z^l + \textnormal{MHA}(\textnormal{Ln}(z^l;\, \gamma_{\text{att}}^l);\, \theta_{\text{att}}^l) \label{eq:hl}\\
  z^{l+1} &= h^l + \textnormal{MLP}(\textnormal{Ln}(h^l;\, \gamma_{\text{mlp}}^l);\, \theta_{\text{mlp}}^l) \label{eq:zl}
\end{align}
The output distribution over the vocabulary is given by:
\begin{align}
  p &= \textnormal{Softmax}(\textnormal{Ln}(z^L;\, \gamma_{\text{final}})\, W_U) \label{eq:softmax}
\end{align}
where $W_U \in \mathbb{R}^{d \times |V|}$ is the unembedding matrix and $L$ is the number of layers.
We define $W_U' := P^\intercal W_U$, $W_E' := W_E P$, $\gamma_{\text{att}}^{l\prime} := \gamma_{\text{att}}^l P$, $\gamma_{\text{mlp}}^{l\prime} := \gamma_{\text{mlp}}^l P$, and $\gamma_{\text{final}}' := \gamma_{\text{final}} P$. Putting all this together
gives the following theorem.\\

\begin{restatable}{theorem}{thmTransformerInvariance}
\label{thm:transformer-invariance}
  For an $L$ layer transformer with parameters
  \[
    \Theta = \bigl(W_E,\, \{\theta_{\text{att}}^l,\, \gamma_{\text{att}}^l,\, \theta_{\text{mlp}}^l,\, \gamma_{\text{mlp}}^l\}_{l=1}^{L},\, \gamma_{\text{final}},\, W_U\bigr),
  \]
  the function computed by the network is identical to that computed by
  \[
    \Theta' = \bigl(W_E',\, \{\theta_{\text{att}}^{l\prime},\, \gamma_{\text{att}}^{l\prime},\, \theta_{\text{mlp}}^{l\prime},\, \gamma_{\text{mlp}}^{l\prime}\}_{l=1}^{L},\, \gamma_{\text{final}}',\, W_U'\bigr).
  \]
  Furthermore, this permutes every parameter matrix in the model.
\end{restatable}

It is also possible to permute the internal dimension of Gated MLPs, with each layer
having a distinct permutation matrix.
Let $Q \in \mathbb{R}^{d_{\text{ff}} \times d_{\text{ff}}}$ be a permutation matrix and define $W_{\text{gate}}' := W_{\text{gate}} Q$,
$W_{\text{up}}' := W_{\text{up}} Q$ and $W_{\text{down}}' := Q^\intercal W_{\text{down}}$.
We write $\theta'_{\text{mlp}} = \{W_{\text{gate}}', W_{\text{up}}', W_{\text{down}}'\}$.\\

\begin{restatable}{proposition}{propGatedMLPInner}
\label{prop:gatedmlp-inner}
  For gated MLPs, we have the following invariance: $\textnormal{MLP}(X;\, \theta'_{\text{mlp}}) = \textnormal{MLP}(X;\, \theta_{\text{mlp}})$
\end{restatable}

Unlike the permutation of the embedding matrix, where the same permutation and its
inverse must be applied throughout the network, different permutations can be applied
to the internal dimensions of the MLPs at each layer (and to different experts in
Mixture of Experts (MoE) models).

As we will see shortly, this means that while the
embedding permutation is the natural choice for removing existing malware implementations,
permutation of the internal dimensions gives significantly more storage capacity when
the goal is encoding bits into the model.

While our method generalizes naturally between dense and mixture-of-experts transformer
models, one caveat is that it does not cover all bias vectors in models which have bias
terms in self-attention layers or in the router networks of mixture-of-experts, with OpenAI's GPT-OSS models
being notable examples of both of these.
For these we further implement permutations
of the order of experts in a layer and of key-value heads in the self-attention mechanisms,
as discussed in Appendix~\ref{app:gpt-oss}.

\subsection{Model Capacity}
To encode information into our LLM, we find the matrix $P$ which encodes
the message that we want. To do this, we first find the permutation matrix $P_{\text{can}}$
which is the permutation matrix such that the columns of $W_E P_{\text{can}}$
are in canonical order. We then find the matrix $P'$ which is the matrix encoding
the $n$th permutation in our ranking, where $n$ is the bits that we wish to encode.
We then apply the permutation $P := P_{\text{can}} P'$ to $W_E$, as well as applying $P$
and $P^{\intercal}$ to the other matrices in our model as dictated by Theorem~\ref{thm:transformer-invariance}.

To decode the message, we again examine the columns of the matrix $W_E'$, this time,
determining the permutation that would be required to get the matrix from
canonical order into its current order. Finding the ranking of this permutation gives us
the encoded bits. Pseudocode for both the encoding and decoding procedures is given in
Appendix~\ref{app:encode-alg}.

It is possible to permute both the columns and rows of a matrix and encode information
in both permutations. Since permuting the rows of a matrix does not change the entries
in a column, only their order, if the hashing function is invariant to that order,
then establishing the ordering of the columns is not required to establish the
ordering of the rows (or vice-versa).

Table~\ref{tab:capacity} shows the embedding dimension and resulting embeddable bit capacity
for several models of interest, computed using the lower bound of Proposition~\ref{prop:encodable-bits}.
For context, the real malware samples embedded by MaleficNet~\citep{hitaj_maleficnet_2022} range from
a few kilobytes to megabytes, so these capacities suffice for the smaller payloads on every model
considered and for the majority on our larger targets.

In practice, the full capacity of the model to store information may not be needed. In
order to not require payload-specific data about the length of the payload in the extraction script, it is
possible to use the first $\lceil \log_2 m \rceil$ encoded bits to give the number
of subsequent bits to be decoded, where $m$ is the total number of encodable bits.

Even with roughly a billion more parameters, Llama 3 8B has the
same capacity as Mistral 7B, since the extra parameters are exclusively in the non-permuted dimension of the
 embedding matrix and are used to accommodate a larger vocabulary size.

Despite having around 13 billion fewer parameters than CodeLlama, GPT-OSS has significantly
greater capacity. This is primarily due to having so many experts, though the internal
dimensions of the MLPs are smaller (2,880 vs 22,016), there are significantly more of
them (768 vs 48).

\begin{table}[ht]
\centering
\small
\begin{tabular}{l rr rrrr rr}
\toprule
 & \multicolumn{2}{c}{Embeddings} & \multicolumn{4}{c}{MLPs} & \multicolumn{2}{c}{Total} \\
\cmidrule(lr){2-3} \cmidrule(lr){4-7} \cmidrule(lr){8-9}
Model & $d$ & bits & $d_{\text{ff}}$ & $L$ & $E$ & bits & bits & KB \\
\midrule
TinyLlama     & 2{,}048 & 19{,}580 & 5{,}632  & 22 & --  & 1{,}365{,}166   & 1{,}384{,}746   & \textbf{173.09}      \\
Mistral 7B    & 4{,}096 & 43{,}250 & 14{,}336 & 32 & --  & 5{,}672{,}544   & 5{,}715{,}794   & \textbf{714.47}      \\
Codestral 22B & 6{,}144 & 68{,}465 & 16{,}384 & 56 & --  & 11{,}521{,}832  & 11{,}590{,}297  & \textbf{1{,}448.79}  \\
Llama 3 8B    & 4{,}096 & 43{,}250 & 14{,}336 & 32 & --  & 5{,}672{,}544   & 5{,}715{,}794   & \textbf{714.47}      \\
CodeLlama 34B$^{\dagger}$  & 8{,}192  & 94{,}685  & 22{,}016 & 48  & --  & 13{,}720{,}992  & 13{,}815{,}677  & \textbf{1{,}726.96}  \\
Llama 3.3 70B$^{\dagger}$  & 8{,}192  & 94{,}685  & 28{,}672 & 80  & --  & 30{,}656{,}000  & 30{,}750{,}685  & \textbf{3{,}843.84}  \\
Llama 3.1 405B$^{\dagger}$ & 16{,}384 & 205{,}747 & 53{,}248 & 126 & --  & 95{,}659{,}830  & 95{,}865{,}577  & \textbf{11{,}983.20} \\
GPT-OSS 20B   & 2{,}880 & 28{,}948 & 2{,}880  & 24 & 32  & 22{,}232{,}064  & 22{,}261{,}012  & \textbf{2{,}782.63}  \\
GPT-OSS 120B$^{\dagger}$ & 2{,}880 & 28{,}948 & 2{,}880  & 36 & 128 & 133{,}392{,}384 & 133{,}421{,}332 & \textbf{16{,}677.67} \\
\bottomrule
\end{tabular}
\caption{We show the embedding dimension and embeddable bit capacity per model. Changing the embeddings is most useful for ensuring that
all parameters are permuted; however, permuting the MLP layers' internal dimensions
contributes the vast majority of the capacity of a typical LLM. For mixture-of-experts
models the inner permutation is applied per expert, so MLP capacity scales with the
number of experts $E$ (shown as ``--'' for dense models).
Models marked $^{\dagger}$ are included for illustrative purposes
but were too large to use in our experiments.}
\label{tab:capacity}
\end{table}

\section{Permutation-Based Countermeasures}
\label{sec:countermeasures}

In previous work~\citep{gilkarov_neuperm_2025, torpmann-hagen_defending_2025},
it was shown that permutation of model parameters was enough to
provide neutralization of potential stegomalware. However, \citet{gilkarov_neuperm_2025} 
achieved permutation of only 60\% of LLM parameters, while the other considered
only MLPs and CNNs. By applying a random permutation as described by
Theorem~\ref{thm:transformer-invariance}, our method transforms 100\% of the LLM weights.
This has the effect of neutralizing all known existing stegomalware techniques, as well
 as the malware technique shown in this paper.
Concretely, we improve upon previous work on permutation-based neutralization of stegomalware in
three ways:

\begin{enumerate}
  \item Whereas previous permutation-based neutralization relied only on \emph{local}
  symmetries, confined within individual components or layers, we additionally exploit
  \emph{global} permutation symmetries that propagate consistently across every layer
  of the network.
  \item We use random derangements---permutations in which no element is returned to its
  original position---rather than random permutations.
  Use of random permutations can pose risks if payload is recoverable from only
  one or two rows/columns returning to their original position.
  \item In some cases, we apply multiple permutations to the same parameter
  matrix, showing that this is not only possible, but necessary to counter permutation-based
  stegomalware.
\end{enumerate}

The first of these improvements is the most substantial, giving the ability
to remove stegomalware from all parameters of an LLM. The second improvement is more
nuanced; we show that for small models it prevents failures which can occur
at low but measurable probability.
In Tables \ref{tab:neutralization-coverage} and \ref{tab:permutation-coverage} we compare the effectiveness of our proposed
neutralization scheme against NeuPerm's neutralization scheme for LLMs with
a MaleficNet payload.

\begin{table}[ht]
\centering
\begin{tabular}{lcc}
\toprule
 & NeuPerm & Ours \\
\midrule
Parameter-wise in MLPs/Gated MLPs            & $\checkmark$ & \checkmark \\
Parameter-wise in self-attention  & $\checkmark^{\ast} $   & \checkmark \\
Parameter-wise in embeddings      & $\times$   & \checkmark \\
Permutation-based                 & $\checkmark^\dagger$   & \checkmark \\
\bottomrule
\end{tabular}
\caption{Stegomalware techniques neutralized by NeuPerm compared with our method. ($\ast$) Has a small probability of
failure. ($\dagger$) Small payloads in embeddings may not be neutralized.}
\label{tab:neutralization-coverage}
\end{table}

We verify these results against TinyLlama, encoding payloads using existing stegomalware techniques from the literature: MaleficNet~\cite{hitaj_maleficnet_2022}, sign mapping~\cite{liu_stegonet_2020} and
Least Significant Byte (LSB) substitution~\cite{liu_stegonet_2020}.
All experiments use TinyLlama-1.1B-Chat, loaded in single precision (\texttt{float32}), with
payloads embedded into a single target matrix at a time (a token embedding, an MLP projection,
or an attention query matrix). We report the rate of \emph{successful neutralization}, defined
as a trial in which the payload can no longer be recovered from the permuted model. As MaleficNet embedding is randomized, we average
each MaleficNet rate over 5 distinct embeddings $\times$ 20 neutralization permutations (100 trials);
LSB embedding is deterministic, so we embed once and apply 100 random neutralization
permutations.

\begin{table}[ht]
\centering
\begin{tabular}{l cc cc cc cc}
\toprule
 & \multicolumn{2}{c}{MaleficNet} & \multicolumn{2}{c}{LSB (partial)} & \multicolumn{2}{c}{LSB (full)} & \multicolumn{2}{c}{Sign-mapping} \\
\cmidrule(lr){2-3} \cmidrule(lr){4-5} \cmidrule(lr){6-7} \cmidrule(lr){8-9}
 & NeuPerm & Ours & NeuPerm & Ours & NeuPerm & Ours & NeuPerm & Ours \\
\midrule
Embeddings      & 0\%   & 100\% & 0\% & 100\% & 0\% & 100\% & 0\% & 100\% \\
MLPs            & 100\% & 100\% & 100\% & 100\% & 100\% & 100\% & 100\% & 100\% \\
Self-attention  & 77\%  & 100\% & 77\% & 100\% & 98\% & 100\% & 98\% & 100\% \\
\bottomrule
\end{tabular}
\caption{Neutralization success rate across stegomalware techniques: a small MaleficNet payload, partial LSB substitution, full LSB substitution, and sign-mapping.}
\label{tab:permutation-coverage}
\end{table}

We find that the permutations proposed by NeuPerm are sufficient
to neutralize malware in MLP layers (allowing the extension to gated MLPs).
However, for MaleficNet, we find that successful neutralization of the query matrix occurs
only 77\% of the time.
In Llama models, grouped-query attention is used, where each key-value pair is associated with multiple queries.
NeuPerm works by permuting the order of these key-value pairs (KV heads) and their corresponding query sets.
In cases where it is recovered, the randomly chosen permutation
fixes at least two of the KV heads in place. In TinyLlama, there are only 4 KV
heads per layer, so this occurs in 7 out of the 24 possible permutations.

For LSB, the success of the neutralization depends on the size and structure of the payload.
Unlike MaleficNet, LSB substitution has no error correction, so the payload survives a
permutation only if every weight it occupies is left in place. A small payload occupies just
the first few hundred weights of the target tensor, which all lie within a single permutation
block (row~0 of an MLP matrix, or the first KV-head block of an attention matrix). It therefore
survives whenever that block is mapped to itself: for the query matrix, with 4 KV heads this
is 6 of the $4!=24$ block permutations, giving a 25\% survival rate. A larger payload spans
every block, so only the identity permutation leaves it intact---survival drops to
$1/24\approx4\%$, and neutralization using a random permutation succeeds $\approx 96\%$ of the time.

For sign-mapping, the performance is dependent upon how the payload
is distributed across the parameter matrix.
We chose the situation where the selected
indices are distributed randomly but if an implementation were chosen which
selects the first available parameter with the correct sign, we would see a result
matching LSB (partial).

For larger models, we would expect the survival rate on NeuPerm to drop, as fewer permutations
return enough indices to their original positions to allow payload recovery. However, even
with large models, the number of KV heads is relatively small (for instance, Llama-3.3-70B
has 8), so a small payload confined to the first row of a weight matrix may survive with NeuPerm 12.5\% of the time. While a payload
embedded across all eight KV heads would be destroyed all but one time in 40,320, this risk
is still reduced to zero using the derangements our method proposes.

We can reason that other variants of EvilModel and StegoNet will also be neutralized
by our approach, as sign-mapping, resilience training and value-mapping all rely on data
being in expected indices. Similarly, LSB variants such as fast-substitution will fail
wherever LSB does, since they expect the payload to be in an order which the permutation
disrupts.

In order to generate derangements, we sample random permutations uniformly with
rejection sampling. This creates a uniform distribution over derangements and has
an expected rejection rate of 63\%. This makes sampling more expensive but still a
negligible part of the full cost of neutralization.

For previous stegomalware, permuting just along the embedding dimension of the LLM is
sufficient to ensure removal of stegomalware. Now, PermaNet specifically requires 
that a random permutation is applied along the dimension along which information was
encoded. This means that to guarantee full removal of all published stegomalware,
it is necessary to apply both the permutation on the embedding dimension and permutation
of the internal dimension of MLP layers.

\section{Performance of Permuted Models}
\label{sec:perf-degradation}
\label{sec:experiments}

Theorem~\ref{thm:transformer-invariance} establishes that the family of permutations
used in our method exactly preserves model behavior in real-valued arithmetic. 
In a deployed model, however, parameters and
activations are stored and computed in finite precision, so two models that are
mathematically equivalent at perfect precision will not necessarily produce identical outputs.

The purpose of the experiments described below is to quantify the effect of our proposed permutations by
measuring the difference in output between an unmodified reference model and a permuted candidate.
In order to give intuition for the scale of this divergence, we compare it against the divergence introduced by
some of the mildest post-training quantization schemes, which are widely accepted as retaining model
utility~\citep{dettmers_llm_int8_2022, kurt_which_2026} but are not theoretically lossless.
For the quantized models we
use popular community GGUF and LLM.int8() quantizations of the reference models available on
Hugging Face (repositories listed in Appendix~\ref{app:model-repos}).

Throughout, we draw permutations uniformly at random. This reflects both settings of
interest: a defender's neutralization permutation is sampled at random
(Section~\ref{sec:countermeasures}), while an attacker's message-encoding permutation is,
by construction, indistinguishable from a uniformly random one once the message is hashed
(Section~\ref{sec:detectability}). Measuring the divergence induced by random permutations
therefore characterizes the performance impact in both the attack and defense settings.

\paragraph{Inference protocol}

For each experiment we compare two models: a \emph{reference} model $R$ and a
\emph{candidate} model $C$, which share an architecture and a tokenizer but differ in weight
ordering, numerical precision, or weight quantization.
Both models are evaluated on a shared
corpus drawn from the WikiText-2 test split~\citep{merity_pointer_2017}, sentence-segmented and filtered
to lengths of $20$--$100$ tokens; from this filtered set we sample $N = 500$ sentences
using a fixed seed so that $R$ and $C$ receive identical input tokens.
Each
model is run in a forward pass over the corpus and, at every non-padding token position,
we record the top-$k$ token indices and their associated raw (pre-softmax) logits, with
$K = 1000$.

\paragraph{Metrics}
Let $\mathcal{V}$ denote the model vocabulary and let $\mathcal{P}$ denote the flattened
set of valid (non-padding) token positions across all $N$ sentences in the corpus. For
each $t \in \mathcal{P}$, let $\ell^R_t, \ell^C_t \in \mathbb{R}^{|\mathcal{V}|}$ denote
the next-token logit vectors produced by $R$ and $C$ respectively, with
$P^R_t := \mathrm{softmax}(\ell^R_t)$ and $P^C_t$ defined analogously. We write
$\mathcal{T}^R_t(k) \subseteq \mathcal{V}$ for the indices of the top-$k$ entries of
$\ell^R_t$ (and $\mathcal{T}^C_t(k)$ analogously), abbreviating
$\mathcal{T}^R_t := \mathcal{T}^R_t(K)$ and $\mathcal{T}^C_t := \mathcal{T}^C_t(K)$
for the full stored slice with $K=1000$. Each metric below is computed per-position and then averaged
(or, for $\Delta_{\max}$, maximized) over $\mathcal{P}$.

\emph{KL divergence.} The Kullback--Leibler (KL) divergence
$D_{\mathrm{KL}}(P^R_t \,\Vert\, P^C_t) = \sum_{v \in \mathcal{V}} P^R_t(v) \log \frac{P^R_t(v)}{P^C_t(v)}$
measures how much information is lost by using $C$'s distribution to approximate $R$'s,
and is a standard distance metric for evaluating compressed language
models~\citep{dutta_accuracy_2024}. We approximate the sum over $\mathcal{T}^R_t$,
renormalizing both distributions on that
support and substituting the smallest logit in $\mathcal{T}^C_t$ as a floor for tokens
in $\mathcal{T}^R_t \setminus \mathcal{T}^C_t$.

\emph{Top-$k$ overlap.} The top-$k$ overlap at position $t$ is the indicator
$\left[\mathcal{T}^R_t(k) = \mathcal{T}^C_t(k)\right]$,
which is $1$ when the unordered top-$k$ token sets coincide and $0$ otherwise. We report
the mean overlap over $\mathcal{P}$ for $k \in \{1, 5, 10\}$.

\emph{Top-$k$ Jaccard similarity.} For larger $k$, we report
the per-position Jaccard similarity
$J_t(k) = \frac{\left| \mathcal{T}^R_t(k) \cap \mathcal{T}^C_t(k) \right|}
{\left| \mathcal{T}^R_t(k) \cup \mathcal{T}^C_t(k) \right|}$,
averaged over $\mathcal{P}$, for $k \in \{100, 1000\}$, which measures fractional
set agreement and retains discriminative power at large $k$.

\emph{Max logit difference.} KL divergence weighs each token by $P^R(v)$, so large logit shifts on
low-probability tokens contribute negligibly even when they reflect substantial
discrepancy at the level of individual tokens~\citep{dettmers_llm_int8_2022, xiao_smoothquant_2023, bondarenko_quantizable_2023}.
We therefore report the per-position max logit difference
$\Delta_t = \max_{v \in \mathcal{T}^R_t} \left| \ell^R_t(v) - \ell^C_t(v) \right|$
(using the same top-$k$ floor for absent tokens as in the KL calculation), and report
its maximum $\Delta_{\max} = \max_{t \in \mathcal{P}} \Delta_t$ over $\mathcal{P}$.
This metric tells us the amount of
numerical drift but may not be reflective of a change in model performance, since
softmax is invariant to a scalar value added to or subtracted from all elements of the logit vector.

\paragraph{Empirical divergence of permuted models}

In order to characterize the empirical difference in outputs between models after the application of permutation,
we compare each permuted candidate against its unmodified reference under the
inference protocol described above, at the models' native \texttt{bf16} precision, across
five independent trials with different random permutations.
For each trial, we apply either the permutation in
Theorem~\ref{thm:transformer-invariance} to the token embeddings and residual stream
(\texttt{embeddings}), the inner-dimension permutation of
Proposition~\ref{prop:gatedmlp-inner} to the gated MLP blocks (\texttt{mlps}), or both
jointly (\texttt{both}).

We report results for TinyLlama-1.1B-Chat-v1.0
(Table~\ref{tab:perm-trials-tinyllama}), Mistral-7B-Instruct-v0.3
(Table~\ref{tab:perm-trials-mistral}) and GPT-OSS-20B (Table~\ref{tab:perm-trials-gpt_oss}); the corresponding results for
Codestral-22B (Table~\ref{tab:perm-trials-codestral}) and
Meta-Llama-3-8B-Instruct (Table~\ref{tab:perm-trials-llama3}) are deferred to
Appendix~\ref{app:additional-results}.
For each statistic we report both the mean
across the five trials and the worst-case trial value, where ``worst'' is the trial
value furthest in the unfavorable direction for that statistic (as indicated by the
arrows in the table headers).

Across all models the divergence is generally largest when both component groups are permuted
(\texttt{both}), but on every metric the difference between the average and worst case
across the five trials is tight, which indicates that the observed divergence does not depend
strongly on the particular permutation chosen.

\begin{table}[ht]
\centering
\small
\begin{tabular}{l cc cc cc}
\toprule
 & \multicolumn{2}{c}{\texttt{embeddings}} & \multicolumn{2}{c}{\texttt{mlps}} & \multicolumn{2}{c}{\texttt{both}} \\
\cmidrule(lr){2-3} \cmidrule(lr){4-5} \cmidrule(lr){6-7}
 & mean & worst & mean & worst & mean & worst \\
\midrule
KL divergence ($\times 10^{-3}$)~$\downarrow$ & $0.555$ & $0.576$ & $0.539$ & $0.541$ & $0.563$ & $0.580$ \\
Top-$1$ overlap (\%)~$\uparrow$ & $98.71$ & $98.61$ & $98.67$ & $98.64$ & $98.75$ & $98.69$ \\
Top-$5$ overlap (\%)~$\uparrow$ & $93.02$ & $92.75$ & $93.04$ & $92.92$ & $93.04$ & $92.80$ \\
Top-$10$ overlap (\%)~$\uparrow$ & $87.26$ & $86.92$ & $87.46$ & $87.35$ & $87.22$ & $86.88$ \\
Top-$100$ Jaccard (\%)~$\uparrow$ & $97.94$ & $97.94$ & $97.96$ & $97.94$ & $97.94$ & $97.92$ \\
Top-$1000$ Jaccard (\%)~$\uparrow$ & $98.43$ & $98.42$ & $98.44$ & $98.43$ & $98.43$ & $98.42$ \\
$\Delta_{\max}$ (logits)~$\downarrow$ & $1.419$ & $1.438$ & $1.625$ & $1.625$ & $1.497$ & $1.594$ \\
\bottomrule
\end{tabular}
\caption{Per-component permutation divergence for TinyLlama-1.1B-Chat-v1.0, under the protocol
described in the text. Arrows indicate the favorable direction for each statistic:
$\uparrow$ ($\downarrow$) means higher (lower) is better.}
\label{tab:perm-trials-tinyllama}
\end{table}

\begin{table}[ht]
\centering
\small
\begin{tabular}{l cc cc cc}
\toprule
 & \multicolumn{2}{c}{\texttt{embeddings}} & \multicolumn{2}{c}{\texttt{mlps}} & \multicolumn{2}{c}{\texttt{both}} \\
\cmidrule(lr){2-3} \cmidrule(lr){4-5} \cmidrule(lr){6-7}
 & mean & worst & mean & worst & mean & worst \\
\midrule
KL divergence ($\times 10^{-3}$)~$\downarrow$             & $0.751$  & $0.921$  & $1.198$  & $1.782$  & $1.763$  & $1.935$  \\
Top-$1$ overlap (\%)~$\uparrow$                           & $98.79$  & $98.75$  & $98.71$  & $98.61$  & $98.67$  & $98.61$  \\
Top-$5$ overlap (\%)~$\uparrow$                           & $93.66$  & $93.49$  & $93.64$  & $93.26$  & $93.22$  & $93.07$  \\
Top-$10$ overlap (\%)~$\uparrow$                          & $88.96$  & $88.73$  & $88.89$  & $88.58$  & $88.26$  & $88.11$  \\
Top-$100$ Jaccard (\%)~$\uparrow$                         & $98.18$  & $98.17$  & $98.19$  & $98.14$  & $98.10$  & $98.08$  \\
Top-$1000$ Jaccard (\%)~$\uparrow$                        & $98.58$  & $98.57$  & $98.58$  & $98.55$  & $98.51$  & $98.50$  \\
$\Delta_{\max}$ (logits)~$\downarrow$                     & $2.644$  & $3.188$  & $3.409$  & $4.438$  & $4.347$  & $5.750$  \\
\bottomrule
\end{tabular}
\caption{Per-component permutation divergence for Mistral-7B-Instruct-v0.3, under the protocol
described in the text. Arrows indicate the favorable direction for each statistic:
$\uparrow$ ($\downarrow$) means higher (lower) is better.}
\label{tab:perm-trials-mistral}
\end{table}

\begin{table}[t]
\centering
\small
\begin{tabular}{l cc cc cc cc}
\toprule
 & \multicolumn{2}{c}{\texttt{embeddings}} & \multicolumn{2}{c}{\texttt{mlps}} & \multicolumn{2}{c}{\texttt{attention}} & \multicolumn{2}{c}{\texttt{all}} \\
\cmidrule(lr){2-3} \cmidrule(lr){4-5} \cmidrule(lr){6-7} \cmidrule(lr){8-9}
 & mean & worst & mean & worst & mean & worst & mean & worst \\
\midrule
KL divergence ($\times 10^{-3}$)~$\downarrow$ & $2.668$ & $2.755$ & $2.102$ & $2.175$ & $2.378$ & $2.447$ & $2.934$ & $3.106$ \\
Top-$1$ overlap (\%)~$\uparrow$ & $97.23$ & $97.16$ & $97.56$ & $97.48$ & $97.38$ & $97.25$ & $97.06$ & $96.88$ \\
Top-$5$ overlap (\%)~$\uparrow$ & $86.00$ & $85.70$ & $87.33$ & $87.22$ & $86.72$ & $86.45$ & $85.14$ & $84.77$ \\
Top-$10$ overlap (\%)~$\uparrow$ & $74.80$ & $74.57$ & $77.18$ & $76.96$ & $75.83$ & $75.53$ & $73.59$ & $73.34$ \\
Top-$100$ Jaccard (\%)~$\uparrow$ & $95.77$ & $95.72$ & $96.25$ & $96.22$ & $96.00$ & $95.97$ & $95.55$ & $95.44$ \\
Top-$1000$ Jaccard (\%)~$\uparrow$ & $96.34$ & $96.29$ & $96.79$ & $96.74$ & $96.55$ & $96.51$ & $96.13$ & $96.03$ \\
$\Delta_{\max}$ (logits)~$\downarrow$ & $4.250$ & $4.812$ & $3.619$ & $3.844$ & $3.906$ & $4.312$ & $4.050$ & $4.500$ \\
\bottomrule
\end{tabular}
\caption{Permutation-invariance trials for \texttt{openai/gpt-oss-20b} (gpt\_oss), 5 seeds. $\uparrow$ ($\downarrow$) means higher (lower) is better.}
\label{tab:perm-trials-gpt_oss}
\end{table}

\paragraph{Full Payload Neutralization}

To verify that our payload calculations are accurate, we perform an end-to-end test
of PermaNet, encoding a payload using a model's full capacity, measuring any degradation
in performance, neutralizing the model with our full-neutralization tool then confirming
that the payload is indeed no longer extractable.

For this test, we use Mistral 7B, and encode an animated GIF from Wikimedia
Commons\footnote{\url{https://commons.wikimedia.org/wiki/File:8-cell-simple.gif}} ($703{,}718$ bytes), padded with a $10{,}752$-byte comment to reach
the model's full $714{,}470$-byte capacity (Table~\ref{tab:capacity}). As expected, we are able to retrieve the image file
perfectly from a safetensors file with the permutation applied, and once the neutralization
is applied, the payload is no longer present.

In Table~\ref{tab:embed-neutralise-mistral}, we show the degradation in performance for both the infected and neutralized
models, showing the infected model has performance in the range shown for random permutations
as measured in Table~\ref{tab:perm-trials-mistral}, as does the neutralized model. Since the neutralized model is
achieved by derangement of the infected model, it is not necessarily a derangement of
the original reference model, though we would expect any difference in performance to
be negligible on average.

This method is dependent upon the rows/columns of the matrices used to encode the
information being unique. Though this is not guaranteed mathematically, in practice the
probability of duplicated rows/columns is negligible.

\begin{table}[t]
\centering
\small
\begin{tabular}{l c c}
\toprule
 & \texttt{embedded} & \texttt{neutralized (full)} \\
\midrule
KL divergence ($\times 10^{-3}$)~$\downarrow$ & $1.911$ & $1.570$ \\
Top-$1$ overlap (\%)~$\uparrow$ & $98.70$ & $98.75$ \\
Top-$5$ overlap (\%)~$\uparrow$ & $92.81$ & $93.58$ \\
Top-$10$ overlap (\%)~$\uparrow$ & $88.24$ & $88.51$ \\
Top-$100$ Jaccard (\%)~$\uparrow$ & $98.11$ & $98.12$ \\
Top-$1000$ Jaccard (\%)~$\uparrow$ & $98.50$ & $98.51$ \\
$\Delta_{\max}$ (logits)~$\downarrow$ & $4.125$ & $4.062$ \\
\bottomrule
\end{tabular}
\caption{Embed-and-neutralize metrics for \texttt{mistralai/Mistral-7B-Instruct-v0.3} (mistral). $\uparrow$ ($\downarrow$) means higher (lower) is better. \emph{Embedded}: permutation encoding the GIF payload; \emph{Neutralized (full)}: full-scope derangement applied post-embedding.}
\label{tab:embed-neutralise-mistral}
\end{table}

\subsection{Comparison to Quantization}

We compare three quantization baselines against our permutation on a \texttt{bfloat16}
TinyLlama-1.1B-Chat~\citep{zhang_tinyllama_2024} reference model,
summarized in Table~\ref{tab:perf-degradation}.

The \texttt{q8\_0} and \texttt{q6\_k} columns compare the native \texttt{bfloat16}
reference against GGUF Q8\_0 and Q6\_K quantized candidates of the same
checkpoint\footnote[1]{Taken from TheBloke/TinyLlama-1.1B-Chat-v1.0-GGUF}, and \texttt{int8}
against an LLM.int8() candidate.
The \texttt{permute} column reports the worst of five random embedding and MLP permutations
(the \texttt{both} regime of Table~\ref{tab:perm-trials-tinyllama}); both reference and
candidate are run at the native \texttt{bf16} precision, so its divergence can be attributed
to the permutation itself.

On TinyLlama, permutation produces measurable but consistently
smaller divergence than every quantization baseline, even taking its worst trial of the five.
The \texttt{permute} column registers a KL divergence of $\approx 5.8 \times 10^{-4}$, retains
top-$1$ token agreement at $98.69\%$ of positions, and exhibits a worst-case logit deviation of
$\Delta_{\max} \approx 1.6$. By comparison, Q8\_0 produces a KL
divergence more than three times larger ($2.094 \times 10^{-3}$) and more than double the
worst-case logit displacement ($3.875$), while Q6\_K and int8 are worse still by one to two
orders of magnitude. Permutation therefore sits below all three quantization schemes on every
metric we report.

We note that the substantial $\Delta_{\max}$ values for Q6\_K ($9.062$
on TinyLlama, and $14.688$ on Mistral as reported below) are consistent with the well-documented
phenomenon of outlier features in transformers, in which a small fraction of activations
or logits dominate model behavior and are correspondingly the most sensitive to coarse
weight quantization~\citep{dettmers_llm_int8_2022, xiao_smoothquant_2023, bondarenko_quantizable_2023}.

\begin{table}[ht]
\centering
\small
\begin{tabular}{l c c c c}
\toprule
 & \texttt{Q8\_0} & \texttt{Q6\_K} & \texttt{int8} & \texttt{permute} \\
\midrule
KL divergence ($\times 10^{-3}$)~$\downarrow$ & $2.094$ & $21.338$ & $16.346$ & $0.580$ \\
Top-$1$ overlap (\%)~$\uparrow$ & $97.60$ & $94.60$ & $93.22$ & $98.69$ \\
Top-$5$ overlap (\%)~$\uparrow$ & $89.02$ & $75.98$ & $67.73$ & $92.80$ \\
Top-$10$ overlap (\%)~$\uparrow$ & $79.75$ & $60.41$ & $47.12$ & $86.88$ \\
Top-$100$ Jaccard (\%)~$\uparrow$ & $96.49$ & $91.46$ & $88.50$ & $97.92$ \\
Top-$1000$ Jaccard (\%)~$\uparrow$ & $97.09$ & $92.68$ & $89.92$ & $98.42$ \\
$\Delta_{\max}$ (logits)~$\downarrow$ & $3.875$ & $9.062$ & $6.469$ & $1.594$ \\
\bottomrule
\end{tabular}
\caption{Distributional divergence between reference and candidate TinyLlama-1.1B-Chat
models on $500$ WikiText-2 test sentences ($K = 1000$). The reference model is
at \texttt{bf16} precision. Top-$k$ overlap is the
percentage of valid positions at which the unordered top-$k$ token sets are
identical; at $k \in \{100, 1000\}$ we instead report the mean top-$k$ Jaccard
similarity (see the metrics outline above), since the overlap indicator
degenerates to near-zero at those values of $k$. The \texttt{Q8\_0}, \texttt{Q6\_K} and
\texttt{int8} candidates are quantizations of the TinyLlama reference model and serve as
utility-preserving reference points; the \texttt{permute} column reports the worst of five
random embedding-and-MLP (\texttt{both}) permutations from
Table~\ref{tab:perm-trials-tinyllama}, evaluated against the same \texttt{bf16} reference.}
\label{tab:perf-degradation}
\end{table}

To validate these results in a larger-scale model,
we repeat the comparison on a \texttt{bfloat16}
Mistral-7B-Instruct-v0.3 reference model using the same inference protocol
($N = 500$ WikiText-2 sentences, $K = 1000$, seed $42$). Our results are summarized in
Table~\ref{tab:perf-degradation-mistral}.
The \texttt{q8\_0} and \texttt{q6\_k} candidates are GGUF
quantizations\footnote[2]{MaziyarPanahi/Mistral-7B-Instruct-v0.3-GGUF} of the
native \texttt{mistralai/Mistral-7B-Instruct-v0.3} reference, \texttt{int8} is an LLM.int8()
candidate, and \texttt{permute} is again the worst of five random \texttt{both}-regime
permutations (Table~\ref{tab:perm-trials-mistral}).

At the $7$B scale the picture is more nuanced than for TinyLlama. Permutation remains far
below the Q6\_K and int8 baselines on almost every metric, and retains higher top-$k$ token
agreement than even Q8\_0 ($98.61\%$ vs $98.43\%$ at top-$1$). On the distributional
metrics, the worst-case permutation is comparable to rather than below Q8\_0: its KL
divergence ($1.935 \times 10^{-3}$) and worst-case logit deviation ($\Delta_{\max} = 5.750$)
sit above Q8\_0's ($1.031 \times 10^{-3}$ and $2.750$), while the KL divergence
remains an order of magnitude or more below Q6\_K's. Permutation is thus of the same order as the mildest
quantization in routine use, and significantly milder than the coarser schemes.

\begin{table}[ht]
\centering
\small
\begin{tabular}{l c c c c}
\toprule
 & \texttt{Q8\_0} & \texttt{Q6\_K} & \texttt{int8} & \texttt{permute} \\
\midrule
KL divergence ($\times 10^{-3}$)~$\downarrow$ & $1.031$ & $102.814$ & $9.438$ & $1.935$ \\
Top-$1$ overlap (\%)~$\uparrow$ & $98.43$ & $95.11$ & $94.93$ & $98.61$ \\
Top-$5$ overlap (\%)~$\uparrow$ & $90.65$ & $81.07$ & $73.79$ & $93.07$ \\
Top-$10$ overlap (\%)~$\uparrow$ & $84.21$ & $67.13$ & $55.91$ & $88.11$ \\
Top-$100$ Jaccard (\%)~$\uparrow$ & $97.33$ & $93.63$ & $91.23$ & $98.08$ \\
Top-$1000$ Jaccard (\%)~$\uparrow$ & $97.78$ & $94.24$ & $92.30$ & $98.50$ \\
$\Delta_{\max}$ (logits)~$\downarrow$ & $2.750$ & $14.688$ & $6.367$ & $5.750$ \\
\bottomrule
\end{tabular}
\caption{Distributional divergence for Mistral-7B-Instruct-v0.3 under the same protocol
as Table~\ref{tab:perf-degradation} ($N = 500$, $K = 1000$). The reference model is
at \texttt{bf16} precision. The \texttt{Q8\_0} and \texttt{Q6\_K} candidates are taken from
MaziyarPanahi/Mistral-7B-Instruct-v0.3-GGUF and \texttt{int8} is an LLM.int8() candidate.
The \texttt{permute} column reports the worst of five random embedding-and-MLP (\texttt{both})
permutations from Table~\ref{tab:perm-trials-mistral}, at Mistral's native \texttt{bf16} precision.}
\label{tab:perf-degradation-mistral}
\end{table}

\paragraph{Summary}
Across both the $1.1$B and $7$B scales, permutation produces distributional divergence far
below the Q6\_K and int8 baselines and on the same order as the mildest scheme, Q8\_0, below
it on every metric for TinyLlama, and comparable to it for Mistral. Our permutation-invariance trials
extend this to larger models: Codestral-22B (Table~\ref{tab:perm-trials-codestral}) and the
$20$B mixture-of-experts \texttt{gpt-oss-20b} (Table~\ref{tab:perm-trials-gpt_oss}). In both,
every permutation regime leaves the output distribution essentially unchanged.

\subsection{Computation at Higher Precisions}

We claim that the divergence in model behavior from the reference model is consistent with
the accumulation of numerical error. We verify that this is the case by examining
the effect of performing the same experiment at higher precisions.
At higher precisions, we would expect less accumulation of numerical error and therefore the divergence would be lower.
For the TinyLlama model, we conduct the same experiment again, but loading the weights
at \emph{float32} and \emph{float64} precisions, comparing the result to performing the
experiment at the standard precision of \emph{bfloat16}. The results of this experiment
are shown in Table~\ref{tab:precision-embeddings_and_mlps}.

\begin{table}[t]
\centering
\small
\begin{tabular}{l cc cc cc}
\toprule
 & \multicolumn{2}{c}{\texttt{bf16}} & \multicolumn{2}{c}{\texttt{f32}} & \multicolumn{2}{c}{\texttt{f64}} \\
\cmidrule(lr){2-3} \cmidrule(lr){4-5} \cmidrule(lr){6-7}
 & mean & worst & mean & worst & mean & worst \\
\midrule
KL divergence~$\downarrow$ & $4.96\times10^{-4}$ & $5.33\times10^{-4}$ & $5.97\times10^{-11}$ & $6.22\times10^{-11}$ & $1.63\times10^{-13}$ & $1.90\times10^{-13}$ \\
Top-$1$ overlap (\%)~$\uparrow$ & $98.82$ & $98.76$ & $100.00$ & $100.00$ & $100.00$ & $100.00$ \\
Top-$5$ overlap (\%)~$\uparrow$ & $93.02$ & $92.91$ & $100.00$ & $99.99$ & $100.00$ & $100.00$ \\
Top-$10$ overlap (\%)~$\uparrow$ & $87.80$ & $87.49$ & $100.00$ & $100.00$ & $100.00$ & $100.00$ \\
Top-$100$ Jaccard (\%)~$\uparrow$ & $98.00$ & $98.00$ & $100.00$ & $100.00$ & $100.00$ & $100.00$ \\
Top-$1000$ Jaccard (\%)~$\uparrow$ & $98.49$ & $98.48$ & $100.00$ & $100.00$ & $100.00$ & $100.00$ \\
$\Delta_{\max}$ (logits)~$\downarrow$ & $2.31\times10^{0}$ & $2.62\times10^{0}$ & $5.62\times10^{-4}$ & $6.76\times10^{-4}$ & $4.87\times10^{-5}$ & $5.48\times10^{-5}$ \\
\bottomrule
\end{tabular}
\caption{Permutation-invariance (\texttt{both}) for \texttt{TinyLlama/TinyLlama-1.1B-Chat-v1.0} across precisions, 5 seeds. $\uparrow$ ($\downarrow$) means higher (lower) is better.}
\label{tab:precision-embeddings_and_mlps}
\end{table}

The results strongly suggest that the source of error is the fixed-precision calculations. This
benefits the defender, since if they have the resources, they can run the network
at a higher precision after performing the countermeasures to reduce performance
degradation. An attacker using PermaNet does not have this advantage, since they cannot
choose the precision at which inference is performed.

\section{Detectability}
\label{sec:detectability}

There are two criteria that we can consider when we discuss detectability of permutation-based
stegomalware. Firstly, whether malware embedded in the order of rows and columns of weight
matrices can be detected by anti-virus tools. Secondly, whether it is possible to detect that
a model contains an encoded message at all, even if it is not possible to determine
the content of the message.

For static analysis and detection of malware, antivirus software requires the ability
to read the raw bytes of the payload. Malware analysis tools may also employ packer detection tools
to determine if the payload has been compressed or encrypted, obfuscation detection, and
unpacking tools to extract the raw bytes for analysis~\cite{egele_survey_2008}.
Current antivirus and static-analysis tooling, however, does not perform the extraction
pipeline our scheme requires. To recover the payload, a tool would have to interpret the
weights as tensors, apply the correct hash function for the canonical ordering, compute
the rank of the permutation of the rows, and convert that integer into bytes. To our knowledge, 
no general-purpose scanner or anti-malware tool could carry out these operations, so the
payload is not exposed to static analysis using current tools.

For the second question, of whether an observer can tell if a model contains a hidden message,
regardless of the message content, we consider the information
theoretic measures of steganographic security introduced in~\citet{cachin_information-theoretic_2000}.
In this framework, security is measured by the KL-divergence
$D(P_{\text{cover}} \,\Vert\, P_{\text{stego}})$ between the distribution $P_{\text{stego}}$ of
steganographically modified objects and the distribution $P_{\text{cover}}$ of innocent cover
objects. A stegosystem is \emph{perfectly secure} when this KL-divergence is zero, that is,
when the two distributions are identical, and $\epsilon$-secure when it is at most $\epsilon$.

To the best of our knowledge, there is no bias toward any particular
ordering of rows or columns in a network that has not been tampered with, so we take the cover distribution
to be uniform over the $n!$ possible orderings. Therefore, the question of detectability
reduces to how close the distribution over orderings emitted by an adversary encoding
stegomalware is to this uniform distribution. For this analysis, we assume
that the warden (i.e., the user trying to discover steganographic messages)
does not have access to the reference model, since in that case they
could trivially detect the permutation by direct comparison.

If the hashing function used to determine the canonical order of rows is public or
easily guessable, a warden can recompute the canonical ordering from the delivered
weights, undo it, and rank the residual permutation to recover the encoded integer $m$.
Even if the encoded bytes are whitened to appear pseudorandom, the method still leaks
information about the existence of a payload, due to a
\emph{support gap}: taking the capacity to be $\lfloor \log_2 n! \rfloor$ bits confines $m$
to $[0, 2^{\lfloor \log_2 n! \rfloor})$, so the ranks in
$[2^{\lfloor \log_2 n! \rfloor}, n!)$ are never achievable in our encoding scheme.
Any observed ordering whose rank falls in this gap therefore cannot have been produced by
our encoder, but does occur in benign orderings.
This means the steganographic and cover distributions differ in their support: any model whose rank
falls in the gap is guaranteed to be benign, so a warden can certify such a model as untampered
with no chance of error. Rejecting all models
whose ranks lie below the gap, whether benign or not, is enough to render the system
perfectly \emph{insecure} in Cachin's framework, i.e., it is not $\epsilon$ secure for any finite choice
of $\epsilon$.
If the warden has knowledge
of the hashing function being used, they can insist on only being sent versions of the
model which fall in the support gap, and therefore cannot contain meaningful data, refusing
benign models which fall in the support of $P_{\text{stego}}$ but removing all risk
of accepting malicious ones.

However, if the defender does not know the exact hashing function used, they are not
able to rely on rejecting message-carrying permutations as a defense. If a keyed hashing
function is used with a sufficiently large key (such as HMAC~\cite{bellare2015new} on the multiset of row entries),
without knowledge of that key, all
orderings once again become equally likely, giving a perfectly secure system by Cachin's framework.

Finally, though the degradation in performance is small, it is possible that this drop in
performance can be used to detect this kind of tampering: the original permutation is
likely the optimal one, while the performance with a permutation applied is likely to
be distributed as a random sample across the performances for all permutations. A heuristic
or statistical test might then flag a model whose performance is a typical sample from this
distribution rather than the optimum. We do not pursue whether such a test is practical,
since a defender need not rely on it: applying the neutralizing permutation of
Section~\ref{sec:countermeasures} removes any payload whether or not its presence could be
detected this way.

In summary, a defender cannot rely on detection. Static analysis does not recover the
payload; the support gap reveals a payload only against an attacker who uses a public or
guessable hash, and a keyed hash makes the system perfectly secure in Cachin's sense. The most
dependable response is therefore neutralization: re-permuting a model to a random
order (Section~\ref{sec:countermeasures}) destroys any payload at a negligible cost to performance.

\section{Related Work}
\label{sec:related}

The ability to encode malware into the weights of a neural network was first demonstrated
by~\citet{liu_stegonet_2020}. In this paper, they proposed multiple techniques by which
a payload could be encoded in and extracted from a deep neural network:
LSB substitution, resilience training,
value mapping and sign mapping. In resilience training, fixed locations in the network
are chosen to encode the bits from the payload and the network is re-trained while
fixing these values. In value mapping and sign mapping, the extraction script contains
a list of locations where the bits to reconstruct the payload can be found.
In this case, the extraction script is required to be payload specific.

In~\citet{wang_evilmodel_2022}, variations on LSB substitution are suggested, allowing
a greater number of bytes to be substituted in MSB reservation and half substitution,
as well as proposing fast substitution, where stegomalware can be embedded without deconstructing
existing parameters at the byte level, allowing for faster encoding.

These techniques were improved upon in~\citet{hitaj_maleficnet_2022}, where the authors introduce MaleficNet, a method of encoding
stegomalware into deep neural networks which makes use of spread spectrum encoding and
error correction codes in order to minimize the effect of embedding a payload on model
performance, while making the resulting stegomalware more robust to fine-tuning and
model pruning. This robustness is a benefit of MaleficNet over our method, though it comes
at the cost of changing the real-valued function computed by MaleficNet. For the
attack vector we consider (the adversary providing the model directly), this is not a
concern, but making PermaNet more robust to fine-tuning and quantization may be a promising
direction of future work.

Recently, two papers have suggested the use of permutation symmetry as a method to
neutralize stegomalware in deep neural networks
\citep{gilkarov_neuperm_2025, torpmann-hagen_defending_2025} with the former showing application
to LLMs. Our neutralization method follows the same principles as these papers, though
targeting LLM architectures specifically in a way that ensures full coverage, as well
as highlighting its necessity through our PermaNet attack.

The main invariance used for neutralization and for the PermaNet attack was suggested
in~\citet{fernandez_functional_2024} as a tool for watermarking models. In that work,
no semantic meaning was given to the permutations, nor was the invariance rigorously proven.
In~\citet{ainsworth_git_2022},
permutation symmetries of neural networks are investigated from a different angle,
arguing that in many situations, different initializations
of the same neural network will typically find the same set of solutions under stochastic
gradient descent, up to permutation symmetry.

In terms of the idea of using permutations to hide information,
the closest related work we could find is~\citet{camenisch_steganographic_2007}, which
investigates using the ability to manipulate the order of streaming TCP packets
to hide semantic information.
This idea was expanded in a tutorial by~\citet{montanez_permutation_2021}, which applied
the idea to ordering textual lists. While the scenario in the former is different
enough to not be directly comparable, the latter uses an encoding scheme which is of order
$\mathcal{O}(n^2)$.

\section{Discussion}
\label{sec:discussion}

A weakness of PermaNet is that the extraction process is rather involved,
although this is also true of MaleficNet and to a lesser extent StegoNet and EvilModel.
However, the extraction code does not have to be packaged with the model; it could arrive via
a different delivery mechanism \cite{lin_stegomalware_2015}, where malicious code would be flagged but a script
of mainly tensor manipulations would garner less suspicion.

In StegoNet's sign-mapping and value-mapping approaches, it is possible to not change
the values of any of the weights in the network, but the steganographic decoder
needs to know the location of the indices to extract the signs/values from, and this
will vary depending on the payload, making the decoder \emph{payload} specific \cite{liu_stegonet_2020}. Meanwhile,
MaleficNet and LSB variants do not require payload specific information, but instead
change the mapping of inputs to outputs in the network, even at high precision \cite{hitaj_maleficnet_2022, wang_evilmodel_2022}.
By contrast, PermaNet only changes the network behavior at a level commensurate with
network precision, and the decoder is payload agnostic.
Because the payload is never present as bytes until decoded, no static
scan of the serialized file, whether pickle opcodes, embedded strings, or raw weight bytes
exposes the payload. 

Our approach is not robust to pruning or quantization, due to its reliance on
hashing functions or ordering schemes which would not survive these processes.
We do not consider this to be a major problem for the attack vectors we consider,
since typically users will download models to use as is. It does however reduce the
potential spread of the malware, since quantizations are an increasingly common modification
of models for practical applications on constrained hardware. Despite this, there
could be other more locality-sensitive hashing methods which make the method more robust.
For quantization in particular, we suggest that the structure inherent in quantization
means it may be possible to create a hashing function specifically designed for this
purpose. Furthermore, it may be possible to find a permutation ranking scheme which is 
more robust to errors, though we leave both these avenues to future research.

Defensively, a drawback of our approach is that developing a neutralization tool for a novel
model architecture requires explicit knowledge of the model's architecture.
Nevertheless, automatic detection of these invariances should be possible for a
wide range of models through analysis of their computational graphs.

Use of the neutralization tool is dependent upon the idea that we can safely load the
parameters of a matrix without loading any associated code. This can be done through
the use of a sandbox environment, by separating weight files from malicious files,
or by static modification of files.

There are other invariants that can be used in order to encode more information
into the model, such as those suggested in \citet{fernandez_functional_2024}. We also leave
these to future work.

\section{Conclusion}
\label{sec:conclusion}

In this paper we studied behavior-preserving permutation symmetries in large language
models from both an offensive and a defensive perspective. We showed how these symmetries
can be exploited to embed a hidden payload in a model's weights in a manner that is
theoretically lossless, requires no retraining, and needs no extraction
script specific to the payload. We also showed that the same symmetries provide a countermeasure: applying random
derangements across sets of parameters neutralizes known stegomalware across all
parameters of a model, improving on prior work. Finally, we quantified the performance
impact of these transformations and found it to be minimal. Together these results
highlight that permutation symmetry serves both attackers and defenders, and that defenders should
apply it pre-emptively to models obtained from untrusted sources.

\section*{Acknowledgements}

This work was supported by the Laboratory for AI Security Research (LASR). The views expressed in this paper are those of the authors and do not necessarily reflect the position of LASR or His Majesty's Government.

The authors would like to thank Kate S, Vicky H, and Graham C for their support, feedback and
many useful conversations which helped shape this paper. We would also like to thank
Chris Norman and Yaz Ibrahim for their contributions during the early stages of this
project, and Tom Morgan and Zoe M for their feedback on earlier drafts.

\bibliographystyle{plainnat}
\bibliography{references}

\appendix

\section{Proofs}
\label{app:proofs}

This appendix collects the proofs of the results stated in the main text.

\propEncodableBits*
\begin{proof}
   For the second inequality, we will use an approximation involving the Gamma function, then use the fact that $\Gamma(n)= (n-1)!$
  From~\citet{abramowitz_handbook_1948}, we get that for the natural log
  of the Gamma function:
  \begin{align*}
    \ln \Gamma(n) \approx (n-\frac{1}{2}) \ln n - n + \frac{1}{2} \ln 2 \pi + \frac{1}{12n} - \frac{1}{360n^3} + \frac{1}{1260n^5}\ldots ,
  \end{align*}
  where the error in the approximation is less than the absolute value of the first
  neglected term and of the same sign. From this, using the fact that $\ln n! = \ln (n-1)! + \ln n$
  we are able to derive our inequality:
  \begin{align*}
    \ln n! &= \ln (n-1)! + \ln n = \ln \Gamma (n) + \ln n \\
    &> (n-\frac{1}{2}) \ln n - n + \frac{1}{2} \ln 2 \pi + \frac{1}{12n} - \frac{1}{360n^3} + \ln n\\
    & = n \ln n + \frac{1}{2} \ln n- n + \frac{1}{2} \ln 2 \pi + \frac{1}{12n} - \frac{1}{360n^3}\\
    & = n \ln n - n + \frac{1}{2} \ln 2 \pi n + \frac{1}{12n} - \frac{1}{360n^3}.
  \end{align*}
  Then, using the change of base formula for logarithms and applying the floor function, we get
  \begin{align*}
    \log_2 n! &> \frac{n \ln n - n + \frac{1}{2} \ln 2 \pi n + \frac{1}{12n} - \frac{1}{360n^3}}{\ln 2}\\
    &\geq \left\lfloor \frac{n \ln n - n + \frac{1}{2} \ln 2 \pi n + \frac{1}{12n} - \frac{1}{360n^3}}{\ln 2} \right\rfloor.
  \end{align*}
  For the first inequality, we use
  \begin{align*}
    \ln \Gamma(n) < (n-\tfrac{1}{2}) \ln n - n + \tfrac{1}{2} \ln 2 \pi + \tfrac{1}{12n} - \tfrac{1}{360 n^3} + \tfrac{1}{1260 n^5}.
  \end{align*}
  Adding $\ln n$ and dividing by $\ln 2$ as before gives that for all $n\geq 1$,
  \begin{align*}
    \log_2 n! &< \frac{n \ln n - n + \frac{1}{2} \ln 2 \pi n + \frac{1}{12n} - \frac{1}{360 n^3} + \frac{1}{1260 n^5}}{\ln 2}\\
    &= \frac{n \ln n - n + \frac{1}{2} \ln 2 \pi n + \frac{1}{12n} - \frac{1}{360 n^3}}{\ln 2} + \frac{1}{1260 n^5 \ln 2}\\
    &< \frac{n \ln n - n + \frac{1}{2} \ln 2 \pi n + \frac{1}{12n} - \frac{1}{360 n^3}}{\ln 2} + 1.
  \end{align*}
  Moving the one to the other side, taking the floor and moving the one outside gives:
  \begin{align*}
    \lfloor \log_2 n! \rfloor - 1 &\leq \left\lfloor \frac{n \ln n - n + \frac{1}{2} \ln 2 \pi n + \frac{1}{12n} - \frac{1}{360 n^3}}{\ln 2} \right\rfloor.
  \end{align*}
\end{proof}

\propMHA*
\begin{proof}
  We first show that for attention heads using $\Theta'$ with permuted inputs $XP$ is equivalent to using
  the unpermuted inputs and original parameters
  \begin{align*}
    \textnormal{Att}(XP;\, \theta'_{\text{att}}) &= \textnormal{Softmax}\!\left(\frac{XP P^\intercal W_Q W_K^\intercal P P^\intercal X^\intercal}{\sqrt{d_k}}\right) XP P^\intercal W_V \\
    &= \textnormal{Softmax}\!\left(\frac{X W_Q W_K^\intercal X^\intercal}{\sqrt{d_k}}\right) X W_V = \textnormal{Att}(X;\, \theta_{\text{att}}).
  \end{align*}
  Defining $\textnormal{head}_i' = \textnormal{Att}(XP;\, \theta'^{(i)}_{\text{att}})$, this gives us the result
  \begin{align*}
    \textnormal{MHA}(XP;\, \theta'_{\text{att}}) &= \textnormal{Concat}(\textnormal{head}_1', \ldots, \textnormal{head}_h')\, W_O' \\
    &= \textnormal{Concat}(\textnormal{head}_1, \ldots, \textnormal{head}_h)\, W_O P \\
    &= \textnormal{MHA}(X;\, \theta_{\text{att}})P. \qedhere
  \end{align*}
\end{proof}

\propGatedMLP*
\begin{proof}
  Applying the properties of the permutation matrix and element-wise operations, we
  have:
  \begin{align*}
    \textnormal{MLP}(XP;\, \theta'_{\text{mlp}}) &= \bigl(\sigma(XP W_{\text{gate}}') \odot XP W_{\text{up}}'\bigr) W_{\text{down}}' \\
    &= \bigl(\sigma(XP P^\intercal W_{\text{gate}}) \odot XP P^\intercal W_{\text{up}}\bigr) W_{\text{down}} P \\
    &= \bigl(\sigma(X W_{\text{gate}}) \odot X W_{\text{up}}\bigr) W_{\text{down}} P \\
    &= \textnormal{MLP}(X;\, \theta_{\text{mlp}})P. \qedhere
  \end{align*}
\end{proof}

\propRMSNorm*
\begin{proof}
  Again through properties of the permutation matrix and element-wise operations, we have:
  \begin{align*}
    \textnormal{Ln}(xP;\, \gamma P)
    &= (\gamma P) \odot \frac{xP}{\sqrt{\textnormal{RMS}^2(xP) + \varepsilon}} = (\gamma P) \odot \frac{xP}{\sqrt{\textnormal{RMS}^2(x) + \varepsilon}} \\
    &= \left(\gamma \odot \frac{x}{\sqrt{\textnormal{RMS}^2(x) + \varepsilon}}\right) P = \textnormal{Ln}(x;\, \gamma)\, P. \qedhere
  \end{align*}
\end{proof}

\thmTransformerInvariance*
\begin{proof}
  We can see
  \begin{align*}
    z^{0\prime} &= \mathbf{t} W_E' = \mathbf{t} W_E P = z^0 P
  \end{align*}
  then we can recursively see the effect on layer outputs
  \begin{align*}
    h^{l\prime} &= z^{l\prime} + \textnormal{MHA}(\textnormal{Ln}(z^{l\prime};\, \gamma_{\text{att}}^{l\prime});\, \theta^{l\prime}_{\text{att}}) \\
                &= z^l P + \textnormal{MHA}(\textnormal{Ln}(z^l P;\, \gamma_{\text{att}}^l P);\, \theta^{l\prime}_{\text{att}}) \\
                &= z^l P + \textnormal{MHA}(\textnormal{Ln}(z^l;\, \gamma_{\text{att}}^l)\, P;\, \theta^{l\prime}_{\text{att}}) \\
                &= z^l P + \textnormal{MHA}(\textnormal{Ln}(z^l;\, \gamma_{\text{att}}^l);\, \theta^l_{\text{att}})\, P = h^l P \\[6pt]
    z^{l+1\prime} &= h^{l\prime} + \textnormal{MLP}(\textnormal{Ln}(h^{l\prime};\, \gamma_{\text{mlp}}^{l\prime});\, \theta^{l\prime}_{\text{mlp}}) \\
                  &= h^l P + \textnormal{MLP}(\textnormal{Ln}(h^l P;\, \gamma_{\text{mlp}}^l P);\, \theta^{l\prime}_{\text{mlp}}) \\
                  &= h^l P + \textnormal{MLP}(\textnormal{Ln}(h^l;\, \gamma_{\text{mlp}}^l)\, P;\, \theta^{l\prime}_{\text{mlp}}) \\
                  &= h^l P + \textnormal{MLP}(\textnormal{Ln}(h^l;\, \gamma_{\text{mlp}}^l);\, \theta^l_{\text{mlp}})\, P = z^{l+1} P
  \end{align*}
  and the output
  \begin{align*}
    p &= \textnormal{Softmax}(\textnormal{Ln}(z^{L\prime};\, \gamma_{\text{final}}')\, W_U') \\
      &= \textnormal{Softmax}(\textnormal{Ln}(z^L P;\, \gamma_{\text{final}} P)\, P^\intercal W_U) \\
      &= \textnormal{Softmax}(\textnormal{Ln}(z^L;\, \gamma_{\text{final}})\, P P^\intercal W_U) \\
      &= \textnormal{Softmax}(\textnormal{Ln}(z^L;\, \gamma_{\text{final}})\, W_U).
  \end{align*}
  Hence, the parameters $\Theta$ and $\Theta'$ result in models that are functionally
  equivalent. Furthermore, every parameter matrix in $\Theta'$ is a permuted version
  of a parameter matrix in $\Theta$ and comprises all parameters in the LLM, completing the
  proof.
\end{proof}

\propGatedMLPInner*
\begin{proof}
 Using the properties of permutation matrices and element-wise operations, we get:
  \begin{align*}
    \textnormal{MLP}(X;\, \theta'_{\text{mlp}}) &= \bigl(\sigma(X W_{\text{gate}}') \odot X W_{\text{up}}'\bigr) W_{\text{down}}' \\
    &= \bigl(\sigma(X W_{\text{gate}} Q) \odot X W_{\text{up}} Q\bigr) Q^\intercal W_{\text{down}} \\
    &= \bigl(\sigma(X W_{\text{gate}}) Q \odot X W_{\text{up}} Q\bigr) Q^\intercal W_{\text{down}} \\
    &= \bigl(\sigma(X W_{\text{gate}}) \odot X W_{\text{up}}\bigr) Q Q^\intercal W_{\text{down}} \\
    &= \bigl(\sigma(X W_{\text{gate}}) \odot X W_{\text{up}}\bigr) W_{\text{down}} \\
    &= \textnormal{MLP}(X;\, \theta_{\text{mlp}}). \qedhere
  \end{align*}
\end{proof}

\section{Message-Encoding Permutation}

In this appendix, we provide pseudocode for the process of encoding information into
a set of matrices in a neural network. The process of neutralization is the same,
except that random bits are chosen for the message. This is equivalent to applying
a random permutation.
\label{app:encode-alg}

\begin{algorithm}[ht]
\caption{Encoding a message via weight permutation.}
\label{alg:encode}
\begin{algorithmic}[1]
\Require Target matrix $T$ with permutation dimension $d_T$
\Require List of affected matrix--dimension pairs $\mathcal{M}$
\Require Message $m \in \mathbb{Z}_{\geq 0}$
\Ensure Matrices in $\mathcal{M} \cup \{(T, d_T)\}$ are permuted so that $m$ is recoverable from $T$
\Function{EncodeMessage}{$T,\, d_T,\, \mathcal{M},\, m$}
  \State $n \gets \text{size of } T \text{ along dimension } d_T$
  \State $\sigma_{\text{can}} \gets$ \Call{CanonicalPermutation}{$T, d_T$} \Comment{permutation into canonical order}
  \State $\sigma_m \gets$ \Call{UnrankNonLex}{$n, m$} \Comment{permutation of rank $m$}
  \State $\sigma \gets \sigma_m \circ \sigma_{\text{can}}$
  \ForAll{$(A, d) \in \mathcal{M} \cup \{(T, d_T)\}$}
    \State $A \gets$ permute $A$ along dimension $d$ according to $\sigma$
  \EndFor
\EndFunction
\end{algorithmic}
\end{algorithm}

\begin{algorithm}[ht]
\caption{Decoding a message from a permuted matrix.}
\label{alg:decode}
\begin{algorithmic}[1]
\Require Target matrix $T$ with permutation dimension $d_T$
\Ensure Bit string $\mathbf{b}$ of length $B$ recovered from the ordering of $T$ along $d_T$
\Function{DecodeMessage}{$T,\, d_T$}
  \State $n \gets \text{size of } T \text{ along dimension } d_T$
  \State $B \gets$ \Call{NumEncodableBits}{$n$} \Comment{Stirling-based lower bound on $\lfloor \log_2 n! \rfloor$}
  \State $\sigma_{\text{can}} \gets$ \Call{CanonicalPermutation}{$T, d_T$} \Comment{permutation into canonical order}
  \State $\sigma \gets \sigma_{\text{can}}^{-1}$ \Comment{canonical into received order}
  \State $r \gets$ \Call{RankNonLex}{$\sigma$} \Comment{recovered integer message}
  \State $\mathbf{b} \gets$ binary expansion of $r$, zero-padded on the left to length $B$
  \State \Return $\mathbf{b}$
\EndFunction
\end{algorithmic}
\end{algorithm}

\section{Non-Lexicographical Ranking and Unranking}
\label{app:rank-unrank}

To rank and unrank permutations in linear time we use the non-lexicographical
ranking proposed in~\citep{myrvold_ranking_2001}. Though it is possible to rank and
unrank lexicographic permutations in linear time, the algorithm is slightly more
involved~\citep{arge_linear-time_2007}.
Though we benefited from an implementation of the algorithms from~\citep{myrvold_ranking_2001},
we found that the recursive implementation as presented in the paper led to reaching
the recursion limits in Python for the large values we needed.

Algorithms~\ref{alg:rank} and~\ref{alg:unrank} give the non-recursive procedures
for ranking a permutation and for recovering a permutation from its rank, using
the non-lexicographical ordering of permutations.

These are non-recursive variants of the algorithms proposed in \citep{myrvold_ranking_2001},
making them suitable for implementation in Python even for large permutation sizes.

\begin{algorithm}[ht]
\caption{Rank a permutation (non-lexicographical order).}
\label{alg:rank}
\begin{algorithmic}[1]
\Require Permutation $\pi$ of $\{0, 1, \ldots, n-1\}$ as an array of length $n$
\Ensure Integer rank $r \in \{0, 1, \ldots, n! - 1\}$
\Function{RankNonLex}{$\pi$}
  \State $n \gets \text{length}(\pi)$
  \If{$n = 0$}
    \State \Return $0$
  \EndIf
  \State $p \gets \text{copy of } \pi$
  \State $q \gets$ inverse permutation of $\pi$
  \State $r \gets 0$
  \State $c \gets 1$
  \For{$i \gets n$ \textbf{down to} $1$}
    \State $s \gets p[i - 1]$
    \State $t \gets q[i - 1]$
    \State swap $p[i - 1]$ and $p[t]$
    \State swap $q[i - 1]$ and $q[s]$
    \State $r \gets r + s \cdot c$
    \State $c \gets c \cdot i$
  \EndFor
  \State \Return $r$
\EndFunction
\end{algorithmic}
\end{algorithm}

\begin{algorithm}[ht]
\caption{Unrank a permutation (non-lexicographical order).}
\label{alg:unrank}
\begin{algorithmic}[1]
\Require Number of elements $n \geq 0$, rank $r \in \{0, 1, \ldots, n! - 1\}$
\Ensure Permutation $\pi$ of $\{0, 1, \ldots, n-1\}$ with rank $r$
\Function{UnrankNonLex}{$n, r$}
  \If{$r \geq n!$}
    \State \textbf{error}: rank out of range
  \EndIf
  \State $\pi \gets [0, 1, \ldots, n - 1]$
  \While{$n > 0$}
    \State swap $\pi[n - 1]$ and $\pi[r \bmod n]$
    \State $r \gets \lfloor r / n \rfloor$
    \State $n \gets n - 1$
  \EndWhile
  \State \Return $\pi$
\EndFunction
\end{algorithmic}
\end{algorithm}

\section{GPT-OSS}
\label{app:gpt-oss}

In this appendix, we demonstrate how our method can be augmented to apply to the
GPT-OSS family of models. GPT-OSS replaces the dense gated MLP of 
Theorem~\ref{thm:transformer-invariance} with a
top-$k$ routed mixture-of-experts model. Unusually for a modern architecture, it also employs
bias terms in its self-attention mechanism. The query, key, value, and sink biases are
internally facing (they do not interact with the global embedding space) and so are not
permuted by the standard embedding permutation. The output-projection bias is the exception,
lying in the global embedding space and permuted along with it.

To account for the mixture-of-experts model, we make two adjustments. For neutralization,
we permute the expert order at each layer; this is necessary to allow for permutation of
the router biases. And, whereas permutation of MLP inner dimensions was previously done per layer,
it is now done per expert. This allows for greater encoding capacity. For the grouped query
attention, we perform a derangement on the key-value heads during
neutralization. Since there are only 8 key-value heads per layer, they offer very little
capacity for encoding, so we do not consider them in an offensive capacity.

Of the models we evaluate empirically in this paper, GPT-OSS 20B offers by far the most storage, able to encode
2.78MB of data. Though we do not run any experiments on GPT-OSS 120B, the 120-billion-parameter
variant, it would be possible to store 16.68MB of data in its experts.

\subsection{Experts}
In a mixture-of-experts model, each MLP layer consists of $E \in \mathbb{N}$ experts,
of which only a subset of size $k <E$ are used in a given forward pass of the network.
The subset used in a particular pass is determined by a router network; in the case of GPT-OSS,
this is a linear layer followed by a softmax operation.
We write that linear layer as $W_R x + b_R$, with
$W_R \in \mathbb{R}^{E \times d}$ and $b_R \in \mathbb{R}^E$.

For a given input embedding $x \in \mathbb{R}^d$ top-$k$ routing selects a subset of these experts
$\mathcal{S}(x) \subseteq \{1, \ldots, E\}$ with mixture weights
$\{\alpha_e(x)\}_{e \in \mathcal{S}(x)}$.
$\alpha_e(x)$ is the softmax over the $k$ largest router logits, with the non-selected experts
assigned weight zero. Each expert $e$ is a gated MLP with parameters
$\theta^e_{\text{mlp}} = \{W^e_{\text{gate}}, W^e_{\text{up}}, W^e_{\text{down}}\}$,
each with associated biases.
Combining the mixture weights and individual expert MLP, the whole MoE block therefore computes
$\textnormal{MoE}(x) = \sum_{e \in \mathcal{S}(x)} \alpha_e(x)\, \textnormal{MLP}(x; \theta^e_{\text{mlp}})$.
Stacking the per-expert weights along a leading expert axis yields tensors
$W^{\bullet}_{\text{gate}}, W^{\bullet}_{\text{up}}, W^{\bullet}_{\text{down}}$ (with $e$th slices
$W^e_{\text{gate}}, W^e_{\text{up}}, W^e_{\text{down}}$).

Let $P_\pi \in \mathbb{R}^{E \times E}$ be the permutation matrix of a permutation $\pi$ of
$\{1, \ldots, E\}$. We define the permuted router weights $W_R' := P_\pi W_R$ and $b_R' := P_\pi b_R$,
and permute $W^{\bullet}_{\text{gate}}, W^{\bullet}_{\text{up}}, W^{\bullet}_{\text{down}}$ and their
biases by $P_\pi$ along the leading (expert) axis. We write $\mathcal{S}'(x)$, $\alpha'_e(x)$ and
$\textnormal{MoE}'(x)$ for the selected set, mixture weights and output of the resulting permuted model.
Given this setup, we can show that this leaves the network unchanged.\\

\begin{lemma}[Expert permutation]
\label{lem:expert-perm}
With the parameters defined above, $\textnormal{MoE}'(x) = \textnormal{MoE}(x)$ for all $x$.
\end{lemma}
\begin{proof}
The permuted router maps the logits to $W_R' x + b_R' = P_\pi(W_R x + b_R)$, a reordering by
$\pi$, so top-$k$ selection and the softmax give $\mathcal{S}'(x) = \pi(\mathcal{S}(x))$ and
$\alpha'_{\pi(e)}(x) = \alpha_e(x)$. Since the expert tensors are addressed by the same index,
the selected experts compute the same $\textnormal{MLP}(x; \theta^e_{\text{mlp}})$ as before, so
the weighted sum $\textnormal{MoE}'(x)$ is preserved.
\end{proof}

\subsection{KV Heads}

Let $H$ denote the number of query heads, $H_{kv}$ the number of key-value heads,
$r = H / H_{kv}$, and $d_h$ the head dimension.
Query heads are partitioned into $H_{kv}$ groups of $r$; within a group $g$ all queries share
the key and value of head $g$. We view
\begin{align*}
W_Q &\in \mathbb{R}^{H_{kv} \times r d_h \times d}, \\
W_K, W_V &\in \mathbb{R}^{H_{kv} \times d_h \times d}, \\
W_O &\in \mathbb{R}^{d \times H_{kv} \times r d_h},
\end{align*}
and associated per-head sinks
$s \in \mathbb{R}^{H_{kv} \times r}$. For a group $g \in \{1, \ldots, H_{kv}\}$, we write
$W_{Q,g} \in \mathbb{R}^{r d_h \times d}$, $W_{O,g} \in \mathbb{R}^{d \times r d_h}$ and
$s_g \in \mathbb{R}^{r}$ for the corresponding blocks.

For full neutralization, we apply two permutations to these heads: an outer permutation
$\rho$ of the $H_{kv}$ groups and, within each group $g$, a permutation $\tau_g$ of its $r$
query heads, with permutation matrices $P_\rho \in \mathbb{R}^{H_{kv} \times H_{kv}}$ and
$P_{\tau_g} \in \mathbb{R}^{r \times r}$. Primed weights denote the result of applying these
matrices to the relevant axes, with $W_O$ receiving the matching inverse permutations on its
input axis so the relabeling is undone at the output. This is formalized in
Lemma~\ref{lem:kv-head-perm}.\\

\begin{lemma}[KV-head permutation]
\label{lem:kv-head-perm}
Permuting the key-value group index of $W_Q, W_K, W_V, W_O, s$ and the biases $b_Q, b_K, b_V$ by
$P_\rho$ leaves attention unchanged. Independently, for each $g \in \{1, \ldots, H_{kv}\}$,
permuting the within-group head index of $W_{Q,g}, W_{O,g}, s_g$ (and $b_{Q,g}$ if present) by
$P_{\tau_g}$ leaves attention unchanged.
\end{lemma}
\begin{proof}
Applying $P_\rho$ to the group axis relabels KV groups consistently across
$W_Q, W_K, W_V, W_O, s$, so that each group of $r$ query heads remains paired with the same key
and value, and the matching inverse on the group index of $W_O$ undoes the relabeling at the
output. The inner permutation $P_{\tau_g}$ acts only on the $r$ query heads sharing the
key/value of group $g$; the matching permutation of $W_{O,g}$ inverts it at the output.
\end{proof}

RoPE is intra-head, so it is unaffected by the permutations above.\\

\begin{theorem}[Complete neutralization for GPT-OSS]
\label{thm:gpt-oss-complete}
Let $\Theta$ be the parameters of a GPT-OSS model. Drawing independently
\begin{enumerate}
  \item a permutation $P$ of $\{1, \ldots, d\}$ applied to the hidden dimension per
  Theorem~\ref{thm:transformer-invariance}, with the gated-MLP step of that theorem
  applied per expert to each $\theta^{l, e}_{\text{mlp}}$;
  \item per layer $l$, a permutation $\pi^l$ of $\{1, \ldots, E\}$ applied as in
  Lemma~\ref{lem:expert-perm};
  \item per layer $l$, a permutation $\rho^l$ of $\{1, \ldots, H_{kv}\}$ and permutations
  $\{\tau^l_g\}_{g=1}^{H_{kv}}$ of $\{1, \ldots, r\}$ applied as in
  Lemma~\ref{lem:kv-head-perm};
  \item per (layer, expert) $(l, e)$, a permutation $Q^{l, e}$ of
  $\{1, \ldots, d_{\text{ff}}\}$ applied to the intermediate dimension of expert
  $\theta^{l, e}_{\text{mlp}}$ by Proposition~\ref{prop:gatedmlp-inner};
\end{enumerate}
yields parameters $\Theta'$ that compute the same function as $\Theta$, and every
parameter tensor of $\Theta$ is moved by at least one of these permutations.
\end{theorem}
\begin{proof}
We have that each permutation works on a different axis of any tensors that they have in common:
$P$ on the hidden/residual stream axis, $\pi^l$ on the expert
axis, $\rho^l$ and $\tau^l_g$ on the head axes, $Q^{l, e}$ on the intermediate dimension axis. The
cited results establish invariance on each axis independently, so their composition is
invariant.

For coverage:
\begin{itemize}
  \item token and output embeddings, both per-layer normalization gains and the final
  normalization gain are moved by $P$;
  \item $W_Q, W_K, W_V, W_O$ are moved by $P$ and $\rho^l$; the biases $b_Q, b_K, b_V$ by
  $\rho^l$ and $b_O$ by $P$;
  \item the sinks $s$ are moved by $\rho^l$ and $\tau^l_g$;
  \item $W_R$ is moved by $P$ and $\pi^l$, and $b_R$ by $\pi^l$;
  \item each $W^{l, e}_{\text{gate}}$, $W^{l, e}_{\text{up}}$ is moved by $P$, $\pi^l$ and
  $Q^{l, e}$, and their biases by $\pi^l$ and $Q^{l, e}$;
  \item each $W^{l, e}_{\text{down}}$ is moved by $P$, $\pi^l$ and $Q^{l, e}$, and its bias by
  $P$ and $\pi^l$.
\end{itemize}
\end{proof}

\section{Additional Per-Model Results}
\label{app:additional-results}

This appendix collects the per-model results for Codestral-22B and
Meta-Llama-3-8B-Instruct referenced in Section~\ref{sec:perf-degradation}.
Tables~\ref{tab:perm-trials-codestral} and~\ref{tab:perm-trials-llama3} report the
per-component permutation divergence under the protocol described there, while
Tables~\ref{tab:quant-llama3-8b} and~\ref{tab:quant-codestral-22b} report the
corresponding quantization-fidelity baselines.

\begin{table}[ht]
\centering
\small
\begin{tabular}{l cc cc cc}
\toprule
 & \multicolumn{2}{c}{\texttt{embeddings}} & \multicolumn{2}{c}{\texttt{mlps}} & \multicolumn{2}{c}{\texttt{both}} \\
\cmidrule(lr){2-3} \cmidrule(lr){4-5} \cmidrule(lr){6-7}
 & mean & worst & mean & worst & mean & worst \\
\midrule
KL divergence ($\times 10^{-4}$)~$\downarrow$             & $4.216$  & $4.450$  & $4.084$  & $4.140$  & $4.400$  & $4.530$  \\
Top-$1$ overlap (\%)~$\uparrow$                           & $98.77$  & $98.67$  & $98.72$  & $98.67$  & $98.75$  & $98.71$  \\
Top-$5$ overlap (\%)~$\uparrow$                           & $93.64$  & $93.40$  & $93.70$  & $93.60$  & $93.57$  & $93.25$  \\
Top-$10$ overlap (\%)~$\uparrow$                          & $89.16$  & $88.87$  & $89.31$  & $89.17$  & $88.95$  & $88.87$  \\
Top-$100$ Jaccard (\%)~$\uparrow$                         & $98.15$  & $98.13$  & $98.16$  & $98.15$  & $98.12$  & $98.09$  \\
Top-$1000$ Jaccard (\%)~$\uparrow$                        & $98.53$  & $98.52$  & $98.55$  & $98.55$  & $98.50$  & $98.49$  \\
$\Delta_{\max}$ (logits)~$\downarrow$                     & $2.400$  & $2.844$  & $2.373$  & $2.750$  & $2.609$  & $3.781$  \\
\bottomrule
\end{tabular}
\caption{Per-component permutation divergence for Codestral-22B, under the protocol described
in the text. Arrows indicate the favorable direction for each statistic:
$\uparrow$ ($\downarrow$) means higher (lower) is better.}
\label{tab:perm-trials-codestral}
\end{table}

\begin{table}[ht]
\centering
\small
\begin{tabular}{l cc cc cc}
\toprule
 & \multicolumn{2}{c}{\texttt{embeddings}} & \multicolumn{2}{c}{\texttt{mlps}} & \multicolumn{2}{c}{\texttt{both}} \\
\cmidrule(lr){2-3} \cmidrule(lr){4-5} \cmidrule(lr){6-7}
 & mean & worst & mean & worst & mean & worst \\
\midrule
KL divergence ($\times 10^{-3}$)~$\downarrow$ & $0.774$ & $0.789$ & $0.773$ & $0.809$ & $0.846$ & $0.872$ \\
Top-$1$ overlap (\%)~$\uparrow$ & $98.62$ & $98.50$ & $98.72$ & $98.64$ & $98.57$ & $98.49$ \\
Top-$5$ overlap (\%)~$\uparrow$ & $92.23$ & $92.02$ & $92.53$ & $92.40$ & $92.05$ & $91.78$ \\
Top-$10$ overlap (\%)~$\uparrow$ & $86.03$ & $85.75$ & $86.35$ & $86.14$ & $85.81$ & $85.51$ \\
Top-$100$ Jaccard (\%)~$\uparrow$ & $97.54$ & $97.52$ & $97.61$ & $97.60$ & $97.49$ & $97.48$ \\
Top-$1000$ Jaccard (\%)~$\uparrow$ & $97.95$ & $97.95$ & $98.01$ & $97.99$ & $97.90$ & $97.90$ \\
$\Delta_{\max}$ (logits)~$\downarrow$ & $1.916$ & $2.938$ & $2.191$ & $2.562$ & $2.272$ & $3.031$ \\
\bottomrule
\end{tabular}
\caption{Per-component permutation divergence for Meta-Llama-3-8B-Instruct, under the protocol
described in the text. Arrows indicate the favorable direction for each statistic:
$\uparrow$ ($\downarrow$) means higher (lower) is better.}
\label{tab:perm-trials-llama3}
\end{table}

\begin{table}[ht]
\centering
\small
\begin{tabular}{l c c c c}
\toprule
 & \texttt{Q8\_0} & \texttt{Q6\_K} & \texttt{int8} & \texttt{permute} \\
\midrule
KL divergence ($\times 10^{-3}$)~$\downarrow$ & $1.802$ & $16.194$ & $25.431$ & $0.872$ \\
Top-$1$ overlap (\%)~$\uparrow$ & $97.58$ & $95.13$ & $91.40$ & $98.49$ \\
Top-$5$ overlap (\%)~$\uparrow$ & $88.50$ & $77.78$ & $61.49$ & $91.78$ \\
Top-$10$ overlap (\%)~$\uparrow$ & $78.89$ & $61.63$ & $38.50$ & $85.51$ \\
Top-$100$ Jaccard (\%)~$\uparrow$ & $96.22$ & $91.46$ & $85.64$ & $97.48$ \\
Top-$1000$ Jaccard (\%)~$\uparrow$ & $96.60$ & $91.52$ & $86.11$ & $97.90$ \\
$\Delta_{\max}$ (logits)~$\downarrow$ & $3.438$ & $7.766$ & $7.438$ & $3.031$ \\
\bottomrule
\end{tabular}
\caption{Quantization fidelity for \texttt{meta-llama/Meta-Llama-3-8B-Instruct} vs the bf16 reference. $\uparrow$ ($\downarrow$) means higher (lower) is better. The \texttt{permute} column reports the worst of five random \texttt{both}-regime permutations (Table~\ref{tab:perm-trials-llama3}).}
\label{tab:quant-llama3-8b}
\end{table}

\begin{table}[ht]
\centering
\small
\begin{tabular}{l c c c c}
\toprule
 & \texttt{Q8\_0} & \texttt{Q6\_K} & \texttt{int8} & \texttt{permute} \\
\midrule
KL divergence ($\times 10^{-3}$)~$\downarrow$ & $1.021$ & $8.557$ & $5.463$ & $0.453$ \\
Top-$1$ overlap (\%)~$\uparrow$ & $98.33$ & $96.20$ & $96.02$ & $98.71$ \\
Top-$5$ overlap (\%)~$\uparrow$ & $90.67$ & $83.98$ & $79.91$ & $93.25$ \\
Top-$10$ overlap (\%)~$\uparrow$ & $85.33$ & $73.23$ & $65.62$ & $88.87$ \\
Top-$100$ Jaccard (\%)~$\uparrow$ & $97.41$ & $94.65$ & $93.51$ & $98.09$ \\
Top-$1000$ Jaccard (\%)~$\uparrow$ & $97.87$ & $95.36$ & $94.48$ & $98.49$ \\
$\Delta_{\max}$ (logits)~$\downarrow$ & $2.562$ & $5.062$ & $9.016$ & $3.781$ \\
\bottomrule
\end{tabular}
\caption{Quantization fidelity for \texttt{mistralai/Codestral-22B-v0.1} vs the bf16 reference. $\uparrow$ ($\downarrow$) means higher (lower) is better. The \texttt{permute} column reports the worst of five random \texttt{both}-regime permutations (Table~\ref{tab:perm-trials-codestral}).}
\label{tab:quant-codestral-22b}
\end{table}

\clearpage
\section{Model Repositories}
\label{app:model-repos}

The reference (full-precision) checkpoints used throughout our experiments are the
following Hugging Face repositories:
\begin{itemize}
  \item \texttt{TinyLlama/TinyLlama-1.1B-Chat-v1.0}
  \item \texttt{mistralai/Mistral-7B-Instruct-v0.3}
  \item \texttt{meta-llama/Meta-Llama-3-8B-Instruct}
  \item \texttt{openai/gpt-oss-20b}
  \item \texttt{mistralai/Codestral-22B-v0.1}
\end{itemize}

The quantized candidates are the corresponding community GGUF repositories below (the
\texttt{int8} baselines are produced at runtime from the reference checkpoints via
LLM.int8() and have no separate repository):
\begin{itemize}
  \item \texttt{TheBloke/TinyLlama-1.1B-Chat-v1.0-GGUF}
  \item \texttt{MaziyarPanahi/Mistral-7B-Instruct-v0.3-GGUF}
  \item \texttt{bartowski/Meta-Llama-3-8B-Instruct-GGUF}
  \item \texttt{bartowski/Codestral-22B-v0.1-GGUF}
\end{itemize}

\end{document}